\documentclass{article}

\usepackage{etex}
\usepackage{graphicx}
\usepackage{multirow}
\usepackage{latexsym}
\usepackage{amsmath,amssymb,amstext,amsthm}
\usepackage{graphicx}
\usepackage{wasysym}
\usepackage{float}
\usepackage{enumerate}
\usepackage[arrow, matrix, curve]{xy}
\usepackage[justification=centering,singlelinecheck=off]{caption}
\usepackage{color}
\usepackage{url}
\usepackage{longtable}
\usepackage{wrapfig}
\usepackage{sidecap}
\usepackage{comment}
\usepackage{stmaryrd}
\usepackage{tikz}
\usetikzlibrary{arrows,automata,backgrounds,calc,chains,fadings,fit,patterns,positioning,shapes.geometric,shapes.symbols,through,trees,decorations.markings,decorations.text}
\usepackage{proof}

\theoremstyle{plain}
\newtheorem{theorem}{Theorem}

\newtheorem{lemma}[theorem]{Lemma}
\newtheorem{proposition}[theorem]{Proposition}

\theoremstyle{definition}
\newtheorem{definition}[theorem]{Definition}
\newtheorem{remark}[theorem]{Remark}
\newtheorem{example}[theorem]{Example}

\newtheorem{question}[theorem]{Question}

\newtheorem{convention}[theorem]{Convention}

\definecolor{orange}{cmyk}{0,0.47,1,0.13}
\definecolor{jgreen}{rgb}{0.0,0.5,0.0}
\definecolor{chocolate}{rgb}{0.82,0.41,0.12}

\newcommand{\ignore}[1]{}

\renewcommand{\mit}{\mathit}

\newcommand{\mcl}{\mathcal}
\newcommand{\msf}{\mathsf}

\makeatletter
\def\moverlay{\mathpalette\mov@rlay}
\def\mov@rlay#1#2{\leavevmode\vtop{%
   \baselineskip\z@skip \lineskiplimit-\maxdimen
   \ialign{\hfil$#1##$\hfil\cr#2\crcr}}}
\makeatother

\newcommand{\id}[1]{#1}
\newcommand{\sub}[2]{{#1}_{#2}}

\newcommand{\binap}[3]{#2\mathbin{#1}#3}
\newcommand{\funap}[2]{#1(#2)}

\newcommand{\indap}[2]{\funap{\sub{#1}{#2}}}

\newcommand{\bfunap}[3]{\funap{#1}{#2,#3}}

\newcommand{\where}{\mathrel{|}}

\newcommand{\sdefdby}{{:}{=}}
\newcommand{\defdby}{\mathrel{\sdefdby}}

\newcommand{\BNFis}{\mathrel{{:}{:}{=}}}
\newcommand{\BNFor}{\mathrel{|}}

\newcommand{\pairlft}{{\langle}}
\newcommand{\pairrgt}{{\rangle}}
\newcommand{\pairsep}{{,\,}}
\newcommand{\pairstr}[1]{\pairlft#1\pairrgt}
\newcommand{\pair}[2]{\pairstr{#1\pairsep#2}}
\newcommand{\triple}[2]{\pair{#1\pairsep#2}}

\newcommand{\tuple}{\pairstr}

\newcommand{\nat}{\mathbb N}
\newcommand{\pnat}{\mathbb{N}_{{>}0}}

\newcommand{\setemp}{{\varnothing}}
\renewcommand{\emptyset}{\setemp}

\newcommand{\sfunin}{{:}}
\newcommand{\funin}{\mathrel{\sfunin}}

\newcommand{\pto}{\rightharpoonup}

\newcommand{\setexp}[1]{\left\{{#1}\right\}}

\newcommand{\modulo}[2]{\left[{#2}\right]_{#1}}

\newcommand{\floor}[1]{\lfloor#1\rfloor}
\newcommand{\ceil}[1]{\lceil#1\rceil}

\newcommand{\slcm}{\mit{lcm}}
\newcommand{\lcm}{\funap{\slcm}}

\newcommand{\lstlength}[1]{|#1|}
\newcommand{\setsize}[1]{|#1|}

\newcommand{\str}[1]{#1^\omega}

\newcommand{\astr}{\sigma}
\newcommand{\bstr}{\tau}
\newcommand{\cstr}{\upsilon}
\newcommand{\iastr}{\sub{\astr}}

\newcommand{\scons}{\strff{:}}

\newcommand{\cons}{\binap{\scons}}
\newcommand{\strcnsd}[1]{\binap{\scons}}

\newcommand{\zeros}{\con{zeros}}
\newcommand{\ones}{\con{ones}}
\newcommand{\alts}{\strcf{alt}}

\newcommand{\strff}{\msf} 
\newcommand{\strcf}{\msf} 
\newcommand{\datc}{\id}

\newcommand{\mybind}[3]{#1#2.\:#3}
\newcommand{\myex}{\mybind{\exists}}
\newcommand{\myall}{\mybind{\forall}}

\newcommand{\mylam}[2]{\lambda#1.#2}

\definecolor{midgray}{rgb}{.128,.128,.128}

\newcommand{\punc}[1]{\:\text{#1}}

\newcommand{\interpret}[1]{\llbracket#1\rrbracket} 
\newcommand{\sinterpret}{\interpret{{\cdot}}}

\newcommand{\shead}{\msf{hd}}
\newcommand{\head}{\funap{\shead}}
\newcommand{\stail}{\msf{tl}}
\newcommand{\tail}{\funap{\stail}}

\newcommand{\szip}{\strff{zip}}
\newcommand{\szipn}{\sub{\szip}}
\newcommand{\zip}{\bfunap{\szip}}
\newcommand{\zipn}{\indap{\szip}}
\newcommand{\seven}{\strff{even}}
\newcommand{\even}{\funap{\seven}}
\newcommand{\sodd}{\strff{odd}}
\newcommand{\odd}{\funap{\sodd}}

\newcommand{\sspj}{\pi}
\newcommand{\spj}[1]{\sspj_{#1}}

\newcommand{\spjx}[2]{\spj{#1,#2}}
\newcommand{\pjx}[2]{\funap{\spjx{#1}{#2}}}

\newcommand{\sfzip}{\mit{zip}}
\newcommand{\sfzipn}{\sub{\sfzip}}

\newcommand{\fzipn}{\indap{\sfzip}}

\newcommand{\aprg}{F}
\newcommand{\fstep}[1]{f_{\hspace{-0.08em}#1}}
\newcommand{\prglist}{\id}

\newcommand{\cpi}[2]{\mathrm{\Pi}^{#1}_{#2}}
\newcommand{\csig}[2]{\mathrm{\Sigma}^{#1}_{#2}}

\newcommand{\sinfred}{{\to^{\omega}}}
\newcommand{\infred}{\mathrel{\sinfred}}

\newcommand{\mred}{\to^{\ast}}

\newcommand{\anum}{p}
\newcommand{\aden}{q}

\newcommand{\ianum}{\sub{\anum}}
\newcommand{\iaden}{\sub{\aden}}
\newcommand{\iafrac}[1]{\frac{\ianum{#1}}{\iaden{#1}}}

\newcommand{\lt}{<}
\newcommand{\gt}{>}

\newcommand{\aoff}{b}

\newcommand{\aspec}{\mcl{S}}

\newcommand{\bspec}{\mcl{S'}}

\newcommand{\con}{\msf}

\newcommand{\zspec}{zip-speci\-fi\-ca\-tion} 

\newcommand{\munterm}[1]{\mbox{$\mu_n$-term}}

\newcommand{\Znterm}[1]{$\szip$-term}
\newcommand{\Znterms}[1]{$\szip$-terms}

\newcommand{\nf}[1]{{#1}{\downarrow}}

\newcommand{\derivatives}[1]{\funap{\partial_{#1}}}

\newcommand{\nb}{\nobreakdash}

\newcommand{\vars}{\mcl{X}}

\newcommand{\undefd}[1]{#1{\uparrow}}
\newcommand{\defd}[1]{#1{\downarrow}}

\newcommand{\nth}{\funap}
\newcommand{\xcons}{\mathbin{\scons}}

\newcommand{\aalph}{\Delta}

\newcommand{\fin}[1]{\sub{\nat}{<#1}}

\newcommand{\aaut}{A}
\newcommand{\baut}{B}
\newcommand{\caut}{C}

\newcommand{\safunct}{F}
\newcommand{\sbfunct}{G}
\newcommand{\afunct}{\funap{\safunct}}
\newcommand{\bfunct}{\funap{\sbfunct}}

\newcommand{\der}{\derivatives}
\newcommand{\base}{\mcl{B}}

\newcommand{\emptyword}{\varepsilon}

\newcommand{\op}{\gamma}

\newcommand{\gen}[2]{\zeta(#1,#2)}
\newcommand{\genz}[1]{\zeta(#1)}

\newcommand{\obase}[1]{\mcl{O}_{#1}}
\newcommand{\nbase}[1]{\mcl{N}_{#1}}

\newcommand{\dfaoograph}{\ograph}
\newcommand{\ographdfao}[1]{A(#1)}

\newcommand{\aograph}{\mcl{G}}
\newcommand{\ograph}{\funap{\mcl{G}}}

\newcommand{\utail}{\funap{\underline{\stail}}}

\newcommand{\trsobs}{\mcl{R}_k(\aspec)}

\newcommand{\zterms}{\mcl{Z}(\aalph,\mcl{X})}
\newcommand{\zkterms}{\mcl{Z}_k(\aalph,\mcl{X})}

\newcommand{\abs}{P}

\tikzset{state/.style={draw=black,ellipse,inner sep=.6mm,outer sep=.5mm}}
\tikzset{default/.style={->,>=stealth',shorten >=1pt,shorten <= 1pt,auto,node distance=2cm,semithick}}
\tikzset{bl/.style={below left of=#1,yshift=3mm}}
\tikzset{br/.style={below right of=#1,yshift=3mm}}

\tikzset{lbr/.style={below right,inner sep=0.5mm}}
\tikzset{lar/.style={above right,inner sep=0.5mm}}
\tikzset{lbl/.style={below left,inner sep=0.5mm}}
\tikzset{lal/.style={above left,inner sep=0.5mm}}

\tikzset{tloop/.style={out=60,in=120,looseness=5}}
\tikzset{bloop/.style={out=-60,in=-120,looseness=5}}
\tikzset{lloop/.style={out=210,in=150,looseness=5}}
\tikzset{rloop/.style={out=-30,in=30,looseness=5}}

\tikzset{lhead/.style={at=(#1.west),anchor=east,xshift=0cm,inner sep=.5mm}}
\tikzset{rhead/.style={at=(#1.east),anchor=west,xshift=0cm,inner sep=.5mm}}
\tikzset{bhead/.style={at=(#1.south),anchor=north,xshift=0cm,inner sep=.5mm}}

\newcommand{\quadeq}{\quad = \quad}
\newcommand{\set}[1]{\{#1\}}
\newcommand{\nott}{\neg}
\newcommand{\andd}{\wedge}
\newcommand{\iif}{\rightarrow}
\newcommand{\iiff}{\leftrightarrow}
\newcommand{\GG}{\mathcal{G}}
\newcommand{\HH}{\mathcal{H}}
\newcommand{\semantics}[1]{[\![#1]\!]}
\renewcommand{\phi}{\varphi}

\newcommand{\dmd}[1]{\langle#1\rangle}
\newcommand{\bx}[1]{[#1]}
\newcommand{\boxdmd}[1]{\{#1\}}

\newcommand{\kset}{\boldsymbol{k}}

\newcommand{\rv}{\con}
\newcommand{\rootsc}{\Zroot}
\newcommand{\Zvar}{\sub{\rv{X}}}
\newcommand{\Zroot}{\Zvar{0}}

\title{Automatic Sequences and Zip-Specifications}

\author{%
  Clemens Grabmayer
  \and J\"{o}rg Endrullis
  \and Dimitri Hendriks
  \and Jan~Willem~Klop
  \and Lawrence S. Moss
}
 
\begin{document}

\maketitle

\begin{abstract}
  We consider infinite sequences of symbols, also known as streams, 
  and the decidability question for equality of streams defined in a restricted format. 
  This restricted format consists of prefixing a symbol at the head of a stream, 
  of the stream function `zip', and recursion variables. 
  Here `zip' interleaves the elements of two streams in alternating order, 
  starting with the first stream; e.g., for the streams defined by
  \(
    \{\, 
      \con{zeros} = \cons{0}{\con{zeros}} ,\,
      \con{ones}  = \cons{1}{\con{ones}} ,\,
      \con{alt}  = \cons{0}{\cons{1}{\con{alt}}}
    \,\}
  \)
  we have $\zip{\con{zeros}}{\con{ones}} = \con{alt}$. 
  The celebrated Thue--Morse sequence is obtained by the succinct `zip-specification'
  \begin{align*}
    \con{M} &= \cons{\datc{0}}{\con{X}} 
    & \con{X} &= \cons{\datc{1}}{\zip{\con{X}}{\con{Y}}} 
    & \con{Y} &= \cons{\datc{0}}{\zip{\con{Y}}{\con{X}}} 
  \end{align*}
  
  Our analysis of such systems employs both term rewriting and coalgebraic techniques. 
  We establish decidability for these zip\nb-specifications, 
  employing bisimilarity of observation graphs based on a suitably chosen cobasis. 
  The importance of zip\nb-specifications resides in their intimate connection with automatic sequences. 
  The analysis leading to the decidability proof of the `infinite word problem' for zip-specifications, 
  yields a new and simple characterization of automatic sequences. 
  Thus we obtain for the binary $\szip$ that a stream is 2\nb-automatic 
  iff its observation graph using the cobasis $\tuple{\shead,\seven, \sodd}$ is finite. 
  Here $\sodd$ and $\seven$ have the usual recursive definition: 
  $\even{a : s} = a : \odd{s}$, and $\odd{a : s} = \even{s}$.
  The generalization to zip-$k$ specifications and their relation to $k$\nb-automaticity is straightforward. 
  In fact, zip\nb-specifications can be perceived as a term rewriting syntax for automatic sequences. 
  Our study of zip\nb-specifications is placed in an even wider perspective 
  by employing the observation graphs in a dynamic logic setting, 
  leading to an alternative characterization of automatic sequences.
  
  We further obtain a natural extension of the class of automatic sequences, 
  obtained by `zip-mix' specifications that use zips of different arities in one specification (recursion system). 
  The corresponding notion of automaton employs a state-dependent input-alphabet, 
  with a number representation $(n)_A = d_m \dots d_0$ where the base of digit $d_i$ 
  is determined by the automaton $A$ on input $d_{i-1}\dots d_0$.
  
  We also show that equivalence is undecidable for a simple extension of the zip-mix format with projections 
  like $\seven$ and $\sodd$.
  However, it remains open whether zip-mix specifications have a decidable equivalence problem.
\end{abstract}

\section{Introduction}\label{sec:intro}

\noindent 
Infinite sequences of symbols, also called `streams', are a playground
of common interest for logic, computer science (functional programming, 
formal languages, combinatorics on infinite words), mathematics (numerations
and number theory, fractals) and physics (signal processing). 
For logic and theoretical computer science this
interest focuses in particular on unique solvability of systems of recursion 
equations defining streams, on expressivity (what scope does a definition 
or specification format have), and productivity (does a stream specification indeed
unfold to its intended infinite result without stagnation). In addition, there is 
the `infinitary word problem': when do two stream specifications over a first-order 
signature define the same stream? And, is that question decidable?
If not, what is the logical complexity?


Against this 
background, we can now situate 
our present paper. 
In the landscape of streams there are some well-known families,
with automatic sequences~\cite{allo:shal:2003} as a prominent family,
including members such as the Thue--Morse sequence~\cite{allo:shal:1999}. 
Such sequences are defined in first-order signature that includes 
some basic stream functions such as $\shead$ (head), $\stail$ (tail), 
`$:$' (prefixing a symbol to an infinite stream),
$\seven$, $\sodd$; all these are familiar from any functional programming language.

One stream function in particular is frequently used in stream specifications.
This is the $\szip$ function, that `zips' the elements of two streams in alternating order, 
starting with the first stream.
Now there is an elegant definition of the Thue--Morse sequence 
$\con{M}$ using only this function $\szip$, 
next to prefixing an element, and of course recursion variables:
\begin{align}
  \con{M} &= \cons{\datc{0}}{\con{X}} 
  & \con{X} &= \cons{\datc{1}}{\zip{\con{X}}{\con{Y}}} 
  & \con{Y} &= \cons{\datc{0}}{\zip{\con{Y}}{\con{X}}}
  \label{eq:spec:morse}
\end{align}

For general term rewrite systems, stream equality is easily seen to be undecidable~\cite{rosu:2006}, 
just as most interesting properties of streams. 
But by adopting some restrictions in the definitional format,
decidability may hold. 

Thus we consider the problem whether definitions like the one of $\con{M}$, 
using only zip next to prefixing and recursion, are still within the realm of decidability.
Answering this question positively turned out to be rewarding.  
In addition to solving the technical problem, the analysis leading to the 
solution had a useful surprise in petto: 
it entailed a new and simple characterization of the important 
notion of $k$\nb-automaticity of streams. 
(The same `aha-insight' was independently
obtained by Kupke and Rutten, preliminary reported in~\cite{kupk:rutt:2011}.)

The remainder of the paper is devoted to an elaboration of several aspects 
concerning zip\nb-specifications and automaticity.
First, 
we treat a representation of automatic sequences
in a framework of propositional dynamic logic, 
employing cobases and the ensuing observation graphs
(used before for the decidability of equivalence) as the underlying semantics 
for a dynamic logic formula characterizing the automaticity of a stream.
Second, 
we are led to a natural generalization of automatic sequences,
corresponding to mixed zip\nb-specifications that contain zip operators of different arities. 
The corresponding type of automaton
employs a state-dependent alphabet. 
Third, 
we show that stream equality for a slight extension of zip-specifications is $\cpi{0}{1}$;
the latter via a reduction from the halting problem of Fractran programs~\cite{conw:1987}.

%

Let us now describe somewhat informally the key method 
that we employ to solve the equivalence problem for zip-specifications.
To that end, consider the specification~\eqref{eq:spec:morse} above with root variable $\con{M}$.
This specification is productive~%
\cite{{sijt:1989},{endr:grab:hend:2008},{endr:grab:hend:isih:klop:2010}} 
and defines the Thue--Morse sequence:
\begin{align*}
  \con{M} \infred
  \cons{0}{\cons{1}{\cons{1}{\cons{0}{
    \cons{1}{\cons{0}{\cons{0}{\cons{1}{
    \cons{1}{\cons{0}{\cons{0}{\cons{1}{
    \cons{0}{\cons{1}{\cons{1}{\cons{0}{\ldots}}}}}}}}}}}}}}}}{,}
\end{align*}
that is, by repeatedly applying rewrite rules that arise by orienting
the equations for $\con{M}$, $\con{X}$ and $\con{Y}$ from left to right,
$\con{M}$ rewrites in the limit to the Thue--Morse sequence~\cite{allo:shal:1999}.

We will construct so-called `observation graphs'
based on the stream cobasis $\triple{\shead}{\seven}{\sodd}$
where all nodes have a double label: 
inside, a term corresponding to a stream (such as $\con{M}$ and $0 : \con{X}$ in Figure~\ref{fig:ograph:morse})
and outside, the head of that stream.
The nodes have outgoing edges to their $\seven$- and $\sodd$-derivatives.
An example is shown in Figure~\ref{fig:ograph:morse}.

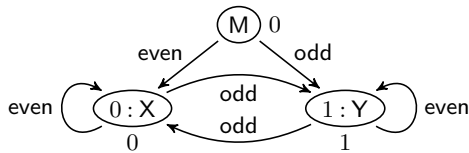
\begin{figure}
  \begin{center}
  \begin{small}
    \begin{tikzpicture}[default]
    \node [state] (q0) {$\con{M}$}; 
    \node (h0) [rhead=q0]{$0$};
    \node [state] (q1) [bl=q0] {$0 \xcons \con{X}$}; 
    \node (h1) [bhead=q1]{$0$};
    \node [state] (q2) [br=q0] {$1 \xcons \con{Y}$}; 
    \node (h2) [bhead=q2] {$1$};
    \path (q0) 
               edge node [lal] {$\seven$} (q1)
               edge node [lar] {$\sodd$} (q2)
          (q1) 
               edge [lloop] node {$\seven$} (q1)
               edge [bend left=25] node [below] {$\sodd$} (q2)
          (q2) 
               edge [rloop] node [right] {$\seven$} (q2)
               edge [bend left=25] node [above] {$\sodd$} (q1);
  \end{tikzpicture}
\vspace{-2ex}
  \end{small}
  \end{center}
  \caption{Observation graph for the specification~\eqref{eq:spec:morse} \\ of the Thue--Morse sequence~$\con{M}$.}
  \label{fig:ograph:morse}
  \vspace{-4ex}
\end{figure}

So, the problem of equivalence of zip-specifications
reduces to the problem of bisimilarity of their observation graphs,
which we prove to be finite.
This does not hold for observation graphs of zip-specifications 
with respect to the cobasis $\pair{\shead}{\stail}$: for this cobasis,
the above specification would yield an infinite observation graph.
(The same would hold for any stream which is not eventually periodic.)

The observation graph in Figure~\ref{fig:ograph:morse} evokes the `aha-insight' mentioned above:
it can be recognized as a DFAO%
  \footnote{%
    The bisimulation collapse of the graph in Fig.~\ref{fig:ograph:morse} 
    identifies the states labeled $\con{M}$ and $0 : \con{X}$, 
    giving rise to the familiar (minimal) DFAO for $\con{M}$.%
  }
(deterministic finite automaton with output)
that witnesses the fact that $\con{M}$ is a $2$\nb-automatic sequence~\cite{allo:shal:2003}.

We will exhibit the close connection between zip\nb-specifications and 
automatic sequences, residing in the coincidence of DFAOs and observation graphs.

\section{Zip-Specifications}\label{sec:k-zip}
For term rewriting notions see further~\cite{tere:2003}.
For $k \in \nat$ we define $\fin{k} = \{0,1,\ldots,k-1\}$. 
Let $\aalph$ be a finite alphabet of at least two symbols, 
and $\mcl{X}$ a finite set of recursion variables.

\begin{definition}\label{def:stream}
  The set $\str{\aalph}$ of \emph{streams over $\aalph$} is defined by
  $\str{\aalph} = \{\astr \where \astr \funin \nat \to \aalph \}$.
\end{definition}

We write $a \xcons \astr$ for the stream $\bstr$ defined by 
$\bstr(0) = a$ and $\bstr(n+1) = \astr(n)$ for all $n\in\nat$.
We define
$\shead : \str{\aalph} \to \aalph$ and
$\stail : \str{\aalph} \to \str{\aalph}$
by
$\head{x \xcons \astr} = x$
and
$\tail{x \xcons \astr} = \astr$.

\begin{convention}
  We usually mix notations for syntax (term rewriting) and semantics (`real' functions).
  Whenever confusion is possible, 
  we use fonts $\mit{fun}$, and $\strff{fun}$ to
  distinguish between functions, and term rewrite symbols, respectively.
\end{convention}

\begin{definition}\label{def:zip}
  For $k \in \pnat$,
  the function $\sfzipn{k} : (\str{\aalph})^k \to \str{\aalph}$
  is defined by the following rewrite rule:
  \begin{align*}
    \zipn{k}{x \xcons \iastr{0},\iastr{1},\ldots,\iastr{k-1}}
    & \to x \xcons \zipn{k}{\sigma_1,\ldots,\sigma_{k-1},\sigma_0}
  \end{align*}
\end{definition}

\noindent
Thus $\sfzipn{k}$ interleaves its argument streams:
\begin{align*}
  \nth{\fzipn{k}{\iastr{0},\ldots,\iastr{k-1}}}{kn+i} & = \nth{\iastr{i}}{n} && (0 \le i \lt k)
\end{align*}

\begin{definition}\label{def:Zterms}
  The set $\zterms$ of \emph{zip\nb-terms} over $\pair{\aalph}{\mcl{X}}$
  is defined by the grammar:
  \begin{align*}
    Z & \BNFis \rv{X} \BNFor a \xcons Z \BNFor \zipn{k}{\underbrace{Z,\ldots,Z}_{\text{$k$ times}}} 
    && (\rv{X} \in \mcl{X},a \in \aalph, k \in \nat)
  \end{align*}
  A \emph{zip-specification~$\aspec$} over $\pair{\aalph}{\mcl{X}}$
  consists of a distinguished variable $\Zroot \in \mcl{X}$ 
  called the \emph{root of $\aspec$},
  and for every $\rv{X} \in \mcl{X}$ 
  a pair $\pair{\rv{X}}{t}$ with $t \in \zterms$ a zip-term.
  We treat these pairs are term rewrite rules,
  and write them as equations $\rv{X} = t$.
%
%
\end{definition}

\begin{definition}
  For $k \in \nat$, the set $\zkterms$ of \emph{zip\nb-$k$ terms}
  is the restriction of $\zterms$ to terms 
  where for every occurrence of a symbol $\szip_{\ell}$ ($\ell \in \nat$)
  it holds that $\ell = k$.
  
  A \emph{zip-$k$ specification} is a zip\nb-specification
  such that for all equations $\rv{X} = t$ it holds that $t \in \zkterms$.
\end{definition}

We always assume for zip-specifications $\aspec$
that every recursion variable is \emph{reachable from the root~$\Zroot$}.

\subsection{Unique Solvability, Productivity and Leftmost Cycles}

\begin{definition}\label{def:model}
  A \emph{valuation} is a mapping $\alpha : \mcl{X} \to \str{\aalph}$. 
  Such a valuation $\alpha$
  extends to $\sinterpret{}_\alpha : \zterms \to \str{\aalph}$ as follows:
  \begin{align*}
    \interpret{\rv{X}}_\alpha & = \alpha(\rv{X}) \\
    \interpret{a:t}_\alpha & = a:\interpret{t}_\alpha \\
    \interpret{\zipn{k}{t_1,\ldots,t_k}}_\alpha & = \fzipn{k}{\interpret{t_1}_\alpha,\ldots,\interpret{t_k}_\alpha} &
  \end{align*}
  A \emph{solution} for a zip-specification $\aspec$ is
  a valuation $\alpha : \mcl{X} \to \str{\aalph}$\!, denoted $\alpha \models \aspec$,
  such that $\interpret{\rv{X}}_\alpha = \interpret{t}_\alpha$
  for all $\rv{X} = t \in \aspec$.
  
  A zip-specification $\aspec$ is \emph{uniquely solvable}
  if there is a unique solution~$\alpha$ for $\aspec$;
  then we let $\sinterpret{}^\aspec = \alpha$ denote this solution.
\end{definition}

\begin{definition}
  Let $\aspec$ and $\bspec$ be zip-specifications with roots~$\rv{X}_0$ and $\rv{X}_0'$, respectively.
  Then $\aspec$ is called \emph{equivalent} to $\bspec$ if
  they have the same set of solutions for their roots:
  \[\{\,\interpret{\Zvar{0}}_\alpha \where \alpha \models \aspec \,\} \;=\; 
    \{\, \interpret{\rv{X}'_0}_{\alpha'} \where \alpha' \models \bspec \,\}\]
\end{definition}

\begin{definition}
  A zip-specification $\aspec$ with root $\Zroot$ is \emph{productive}
  if there exists a reduction of the form 
  $\Zroot \mred a_1 : \ldots : a_n : t$
  for all $n \in \nat$.
  If a zip-specification $\aspec$ is productive, then 
  $\aspec$ is said to \emph{define the stream $\interpret{\Zroot}^\aspec$}
  where $\Zroot$ is the root of $\aspec$.
\end{definition}

Note that if a specification is productive,
then by confluence of orthogonal term rewrite systems~\cite{tere:2003},
there exists a rewrite sequence of length $\omega$
that converges towards an infinite stream term $a_1 : a_2 : a_3 : \ldots$ in the limit.


While productivity is undecidable~\cite{endr:grab:hend:cade:2009,simo:2009} 
for term rewrite systems in general, zip-specifications fall into the class of 
`pure stream specifications'~\cite{endr:grab:hend:2008,endr:grab:hend:isih:klop:2010}
for which (automated) decision procedures exist.
However, the latter would be taking a sledgehammer to crack a nut.
For zip-specifications, productivity boils down to a simple syntactic criterion.

\begin{definition}
  Let $\aspec$ be a zip-specification.
  A \emph{step in $\aspec$} is pair of terms $\pair{s}{t}$, denoted by $s \leadsto t$,
  such that 
  (a) $s \to t \in \aspec$,
  (b) $s = a:t$, or
  (c) $s = \zipn{k}{\ldots,t,\ldots}$.
  A \emph{guard} is a step of form (b).
  A \emph{left-step $s \leadsto_\ell t$ in $\aspec$} 
  is a step $s \leadsto t$ of the form (a), (b) or (c') $s = \zipn{k}{t,\ldots}$.

  A \emph{cycle in $\aspec$} is a sequence $t_1,t_2,\ldots,t_n$
  such that $t_1 = t_n \in \mcl{X}$ and
  $t_i \leadsto t_{i+1}$ for $1 \le i < n$.
  A \emph{leftmost cycle in $\aspec$} is a cycle $t_1,t_2,\ldots,t_n$
  such that $t_i \leadsto_\ell t_{i+1}$ for $1 \le i < n$.
\end{definition}

\begin{example}\label{ex:unprod}
  Consider the following specification\\[-.5ex]
  \begin{minipage}{.73\columnwidth}
  \begin{align*}
    \rv{X} &= \zip{1:\rv{X}}{\rv{Y}} \\
    \rv{Y} &= \zip{\rv{Z}}{\rv{X}} \\
    \rv{Z} &= \zip{\rv{Y}}{0:\rv{Z}}
  \end{align*}
  visualized as the cyclic term graph on the right.
  The leftmost cycle
  $\rv{Y} \leadsto_\ell \zip{\rv{Z}}{\rv{X}} \leadsto_\ell \rv{Z} \leadsto_\ell \zip{\rv{Y}}{0\,{:}\,\rv{Z}} \leadsto_\ell \rv{Y}$
  is not guarded.
  \end{minipage}
  \begin{minipage}{.25\columnwidth}
  \newcommand{\zipfan}[2]{
    \coordinate (#1 out) at (#2);
    \coordinate (#1 south west) at ($(#2) + (-6mm,-6mm)$);
    \coordinate (#1 south east) at ($(#2) + (6mm,-6mm)$);
    \draw (#1 out) -- (#1 south west) -- (#1 south east) -- cycle;
    \node at (#1 out) [yshift=-2.75mm] {$\szip$};
    \coordinate (#1 inL) at ($(#1 south west)!.25!(#1 south east)$);
    \coordinate (#1 inR) at ($(#1 south west)!.75!(#1 south east)$);
  }
  \vspace{.25cm}
  \hspace*{-11mm}
  \begin{tikzpicture}[thick,>=stealth,inner sep=1pt,yscale=.7]
    \zipfan{X}{0mm,0mm} \node at ($(X out) + (0mm,3.5mm)$) {$\rv{X}$}; 
    \zipfan{Y}{$(X inR) + (7mm,-7mm)$} \node at ($(Y out) + (2mm,2mm)$) {$\rv{Y}$};
    \zipfan{Z}{$(Y inL) + (-2mm,-7mm)$} \node at ($(Z out) + (-2mm,2mm)$) {$\rv{Z}$};
    \draw [thin] (Y out) to[out=90,in=-91] (X inR);
    \begin{scope}[->,thin,looseness=2,rounded corners=2mm,shorten >= 1mm]
    \draw (X inL) -- +(0mm,-8mm) node [thick,midway,circle,draw=black,fill=white] {1} to[out=180,in=160] (X out);
    \draw [looseness=1] (Y inR) -- +(0mm,-3mm) to[out=0,in=20] (X out);
    \draw (Z inR) -- +(0mm,-8mm) node [thick,midway,circle,draw=black,fill=white] {0} to[out=0,in=20] (Z out);    
    \end{scope}
    \draw [ultra thick] (Z out) to[out=90,in=-91] (Y inL);
    \draw [ultra thick,->,looseness=1,rounded corners=2mm,shorten >= 1mm] (Z inL) -- +(0mm,-3mm) to[out=-180,in=160] (Y out);
  \end{tikzpicture}
  \end{minipage}
\end{example}

For term rewriting systems in general, productivity implies~the uniqueness of solutions,
but unique solvability is not sufficient for productivity.
For zip-specifications it turns out that both concepts coincide.
Here we need that $\aalph$ is not a singleton ---
otherwise every specification has a unique solution.

\begin{theorem}\label{thm:unique:productive:cycle}
  For zip-specifications $\aspec$ these are equivalent:
  \begin{enumerate}
    \item $\aspec$ is uniquely solvable.
    \item $\aspec$ is productive.
    \item $\aspec$ has a guard on every leftmost cycle.
  \end{enumerate}
\end{theorem}

\subsection{Evolving and Solving Zip-Specifications}

The key to the proof of Theorem~\ref{thm:unique:productive:cycle}
consists of a transformation of zip-specific{-}ations 
by (i) simple equational logic steps, and (ii) internal rewrite steps.
\begin{definition}\label{def:evolve}
  For zip\nb-specifications $\aspec,\bspec$
  we say \emph{$\aspec$ evolves to $\bspec$}, 
  denoted by $\aspec \circlearrowright \bspec$,
  if one of the conditions holds:
  \begin{enumerate}
    \item 
      $\aspec$ contains an equation $\rv{X} = a:t$ with $\rv{X} \ne \Zroot$ and
      $\bspec$ is obtained from $\aspec$ 
      by: let $\rv{X}'$ be fresh and
      \begin{enumerate}
        \item exchange the equation $\rv{X} = a:t$ for $\rv{X}' = t$, then
        \item replace all $\rv{X}$ in all right-hand sides by $a:\rv{X}'$, and
        \item finally rename $\rv{X}'$ to $\rv{X}$ ($\rv{X}$ is no longer used).
      \end{enumerate}
    \item
      $\aspec$ contains an equation $\rv{X} = t$
      such that $t$ rewrites to $t'$ via a zip-rule (Definition~\ref{def:zip}),
      and $\bspec$ is obtained from $\aspec$ by 
      replacing the equation $\rv{X} = t$ with $\rv{X} = t'$.
  \end{enumerate}
\end{definition}

The condition $\rv{X} \ne \Zroot$ in clause~(i) 
guarantees that the meaning (its solution) is preserved under evolving.
It prevents transforming a specification like $\Zroot = 0 : 1 : \Zroot$ into $\Zroot = 1 : 0 : \Zroot$
which clearly has a different solution.

\begin{lemma}
  Let $\aspec \circlearrowright \bspec$. 
  Then for every $\alpha : \mcl{X} \to \str{\aalph}$ it holds that 
  $\alpha$ is a solution of $\aspec$ if and only if $\alpha$ is a solution of $\bspec$.
  Moreover, if $\aspec$ is productive then so is $\bspec$.
\end{lemma}

\begin{definition}
  A zip-specification $\aspec$ is said to have a \emph{free root}
  if the root $\Zroot$ of $\aspec$ does not occur in any right-hand side of~$\aspec$.
\end{definition}

\begin{lemma}\label{lem:free:root}
  Every zip-specification can be transformed into an equivalent one with free root.
\end{lemma}

The following lemma relates rewriting to evolving:
\begin{lemma}
  Let $\aspec$ be a zip-specification with free root~$\Zroot$.
  There exists a reduction
  $\Zroot \mred a_1 : \ldots : a_n : t$ in $\aspec$
  if and only if
  there exists a zip-specification $\aspec'$ such that $\aspec \circlearrowright^* \bspec$
  and $\bspec$ contains an equation of the form $\Zroot = a_1 : \ldots : a_n : t'$.
\end{lemma}

{%
\renewcommand{\xcons}{\,{:}\,}%
\newcommand{\xeq}{\,{=}\,}%
\renewcommand{\space}{\;\;\,}%
\begin{example}\label{ex:evolve}
  We evolve the following specification:
  \begin{gather*}%
    \rv{X} \xeq \zip{1 \xcons \rv{X}}{\rv{Y}} \space
    \rv{Y} \xeq \underline{0 \xcons} \tail{\zip{\rv{Z}}{\rv{X}}} \space
    \rv{Z} \xeq \zip{\rv{Y}}{0 \xcons \rv{Z}} \\
    \rv{X} \xeq \zip{1 \xcons \rv{X}}{\overline{0 \xcons} \rv{Y}} \space
    \rv{Y} \xeq \tail{\zip{\rv{Z}}{\rv{X}}} \space
    \rv{Z} \xeq \zip{\overline{\underline{0 \xcons}} \rv{Y}}{0 \xcons \rv{Z}} \\
    \ldots \space 
    \rv{Y} \xeq \tail{\zip{\rv{Z}}{\rv{X}}} \space
    \rv{Z} \xeq \overline{\underline{0 \xcons}} \zip{0 \xcons \rv{Z}}{\rv{Y}} \\
    \ldots \space 
    \rv{Y} \xeq \tail{\zip{\overline{\underline{0 \xcons}} \rv{Z}}{\rv{X}}} \space
    \rv{Z} \xeq \zip{0 \xcons \overline{0 \xcons} \rv{Z}}{\rv{Y}} \\
    \ldots \space 
    \rv{Y} \xeq \underline{\stail(\overline{0 \xcons}} \zip{\rv{X}}{\rv{Z}}) \space
    \rv{Z} \xeq \zip{0 \xcons 0 \xcons \rv{Z}}{\rv{Y}} \\
    \rv{X} \xeq \zip{1 \xcons \rv{X}}{0 \xcons \rv{Y}} \space
    \rv{Y} \xeq \zip{\rv{X}}{\rv{Z}} \space
    \rv{Z} \xeq \zip{0 \xcons 0 \xcons \rv{Z}}{\rv{Y}}
  \end{gather*}%
  Note that the contracted redexes are underlined and the created symbols are overlined.
  Also note that invoking a free root is not needed for the evolution above. 
\end{example}

Strictly speaking, the last step in the above example 
is not covered by Definition~\ref{def:evolve}
since the rule for `$\stail$' is not included.
We have chosen this example to demonstrate another principle.
The specification we started from is obtained from Example~\ref{ex:unprod}
by inserting $0 : \tail{\ldots}$ on an unguarded leftmost cycle.
Evolving has resulted in a productive zip\nb-specification (now every leftmost cycle is guarded)
that represents a solution of the original specification.
Similarly, by inserting $1 : \tail{\ldots}$, 
we obtain the solution:
\begin{gather*}
  \rv{X} \xeq \zip{1 \xcons \rv{X}}{1 \xcons \rv{Y}} \space
  \rv{Y} \xeq \zip{\rv{X}}{\rv{\rv{Z}}} \space
  \rv{\rv{Z}} \xeq \zip{0 \xcons 1 \xcons \rv{\rv{Z}}}{\rv{Y}}
\end{gather*}
}%
The insertion of $0 : \tail{\ldots}$ and $1 : \tail{\ldots}$
corresponds to choosing whether we are interested in a solution for $\rv{Y}$ starting with head $0$ or $1$.
To see that the result of the insertions are valid solutions
it is crucial to observe that the symbol `$\stail$' in the inserted $a : \tail{\ldots}$
disappears by consuming a `descendant' of the element $a\in \aalph$.
In general we have:
\begin{lemma}\label{lem:solve}
  Let $\aspec$ be a zip-specification.
  Define the set $\{\rv{Y}_1,\ldots,\rv{Y}_m\}$ to contain precisely one recursion variable
  from every unguarded leftmost cycle from $\aspec$.
  
  Let $\vec{a} = \tuple{a_1,\ldots,a_m} \in \aalph^m$ and define $\aspec_{\vec{a}}$ to be obtained from $\aspec$ by 
  replacing each equation $\rv{Y}_i = t_i$ by $\rv{Y}_i = a_i : \tail{t_i}$.
  Subsequently, we can by the evolving procedure eliminate the occurrences 
  of the symbol~$\stail$ as in Example~\ref{ex:evolve}. 
  Then $\aspec_{\vec{a}}$ is productive, and 
  the unique solution $\sinterpret{}^{\aspec_{\vec{a}}} : \mcl{X} \to \str{\aalph}$ 
  is a solution of $\aspec$.
  Hence, $\{\aspec_{\vec{a}} \where \vec{a} \in \aalph^m\}$
  is the set of all solutions of $\aspec$, in particular, 
  $\aspec$ has $\setsize{\aalph}^m$ different solutions.
\end{lemma}

\subsection{Formats of Zip-Specifications}

\begin{definition}\label{def:flat:specs}\normalfont
  A zip-specification~$\aspec$ is called \emph{flat}
  if each of~its equations is of the form: 
  \begin{align*}
    \Zvar{i} = \cons{c_{i,1}}{\cons{\ldots}{\cons{c_{i,m_i}}{\zipn{k_i}{\Zvar{i,1},\ldots,\Zvar{i,k_i}}}}}
    && (0 \le i < n)
  \end{align*}
  for $m_i,k_i \in \nat$, $k_i \ge 2$, recursion variables $\Zvar{i}, \Zvar{i,1}, \ldots, \Zvar{i,k_i}$
  and data constants $c_{i,1}, \ldots, c_{i,m_i}$.
\end{definition}

Zip\nb-free cycles correspond to periodic sequences,
and these can be specified by flat zip-$k$ specifications. 
Together with unfolding and introduction of fresh variables we then obtain:

\begin{lemma}\label{lem:flat}
  Every productive zip-$k$ specification can be transformed 
  into an equivalent productive, flat zip-$k$ specification.
\end{lemma}

\section{Zip-Specifications and Observation Graphs}\label{sec:k-zip:graphs}

For the decidability result and the connection with automaticity
we need to observe streams and compare them. 
This is done with observations in terms of a cobasis 
and bisimulations to compare the resulting graphs.

\subsection{Cobases, Observation Graphs, and Bisimulation}

For general introductions to coalgebra 
we refer to~\cite{barw:moss:1996,sang:rutt:2012}.
We first introduce the notion of `cobasis'~\cite{rosu:2000,kupk:rutt:2010}.
For the sake of simplicity, 
we restrict to the single observation $\shead$.

\begin{definition}\label{def:cobasis}
  A \emph{stream cobasis~$\base = \tuple{\shead,\tuple{\op_1,\ldots,\op_k}}$} is a tuple
  consisting of \emph{operations} $\op_i : \str{\aalph} \to \str{\aalph}$
  ($1 \le i \le k$)
  such that
  for all $\astr,\bstr\in\str{\aalph}$ it holds that $\astr = \bstr$
  whenever
  \begin{align*}
    \head{\op_{i_1}(\ldots(\op_{i_n}(\astr)) \ldots)} = \head{\op_{i_1}(\ldots(\op_{i_n}(\bstr)) \ldots)}
  \end{align*}
  for all $n\in\nat$ and $1 \le i_1,\ldots,i_n \le k$.
\end{definition}

As $\shead$ is integral part of every stream cobasis, we suppress $\shead$ 
and write $\tuple{\op_1,\ldots,\op_k}$ as shorthand for $\tuple{\shead,\tuple{\op_1,\ldots,\op_k}}$.

\begin{definition}
  For $i \in \nat$, $k \in \pnat$
  define $\spjx{i}{k} : \str{\aalph} \to \str{\aalph}$:
  \begin{align*}
    \pjx{0}{k}{x \xcons \astr} 
    & \to x \xcons \pjx{k-1}{k}{\astr} \\
    \pjx{i+1}{k}{x \xcons \astr} 
    & \to \pjx{i}{k}{\astr}
  \end{align*}
  For every $k \ge 2$ we define two stream cobases:
  \begin{align*}
    \nbase{k} & = \tuple{\spjx{0}{k},\dots,\spjx{k-1}{k}} &
    \obase{k}  & = \tuple{\spjx{1}{k},\dots,\spjx{k}{k}} 
  \end{align*}
\end{definition}
Note that $\pjx{i}{k}{\astr}$ selects an arithmetic subsequence of $\astr$; 
it picks every $k$-th element beginning from index $i$:
  $\nth{\pjx{i}{k}{\astr}}{n} = \nth{\astr}{kn + i}$.
The $\spjx{i}{k}$ are generalized $\seven$ and $\sodd$ functions,
in particular we have:
$\stail = \spjx{1}{1}$, $\seven = \spjx{0}{2}$ and $\sodd = \spjx{1}{2}$. 

Observe that $\nbase{k}$ and $\obase{k}$ are cobases,
that is, every element of a stream can be observed.
The main difference between $\nbase{k}$ and $\obase{k}$ is that $\nbase{k}$
has an ambiguity in naming stream entries: $\head{\sigma} = \head{\even{\sigma}}$.
On the other hand, $\obase{k}$ is an orthogonal basis,
names of stream entries are unambiguous.

We employ the following simple coinduction principle:

\begin{definition}\label{def:bisim}
  Let $\base = \tuple{\op_1,\ldots,\op_k}$ be a cobasis.
  A \emph{$\base$\nb-bisimulation} is a 
  relation $R \subseteq \str{\aalph} \times \str{\aalph}$
  s.t.\ $\pair{\astr}{\bstr} \in R$ 
  implies $\head{\astr} = \head{\bstr}$ and 
  $\pair{\op_i(\astr)}{\op_i(\bstr)} \in R$ 
  for $1 \le i \le k$.
\end{definition}

\begin{lemma}\label{lem:coinduction}
  For all $\astr,\bstr\in\str{\aalph}$ 
  it holds that
  $\astr = \bstr$ 
  if and only if 
  there exists a $\base$\nb-bisimulation~$R$
  such that $\pair{\astr}{\bstr} \in R$.
  \qed
\end{lemma}

We now further elaborate the coalgebraic perspective.
The following definition formalizes `$\base$-observation graphs'
where $\base = \tuple{\op_1,\ldots,\op_k}$ is a cobasis.
Every node~$n$ will represent the stream $\interpret{n} \in \str{\aalph}$,
and if the $i$-th outgoing edge of $n$ points to node $m$ then $\op_i(\interpret{n}) = \interpret{m}$. 

\newcommand{\groot}{r}
\begin{definition}\label{def:ograph}
  Let $\base = \pair{\shead}{\tuple{\op_1,\ldots,\op_k}}$ be a stream cobasis,
  and let $\safunct$ be the functor $\afunct{X} = \aalph \times X^k$.
  
  A \emph{$\base$-observation graph} is 
  an $\safunct$-coalgebra $\aograph = \pair{S}{\pair{o}{n}}$ 
  with a distinguished \emph{root} element $\groot \in S$,
  such that there exists an $\safunct$\nb-homomorphism $\sinterpret \funin S \to \str{\aalph}$
  from $\aograph$ to the $\safunct$-coalgebra $\pair{\str{\aalph}}{\base}$ of all streams
  with respect to $\base$:
  
  \begin{center}
    \begin{tikzpicture}[node distance=35mm,>=stealth,thick]
      \node (S) {$S$};
      \node (TS) [below of=S,node distance=12mm] {$\aalph \times S^k$};
      \node (Aw) [right of=S] {$\str{\aalph}$};
      \node (TAw) [right of=TS] {$\aalph \times (\str{\aalph})^k$};
      \draw [->] (S) -- (Aw) node [midway,above] {$\sinterpret$}; 
      \draw [->,shorten >= -1mm] (S) -- (TS) node [midway,left] {$\pair{o}{n}$};
      \draw [->,shorten >= -1mm] (Aw) -- (TAw) node [midway,right] {$\base$}; 
      \draw [->] (TS) -- (TAw) node [midway,above] {$\text{id}\times\sinterpret^k$}; 
    \end{tikzpicture}
  \end{center}
  The observation graph $\aograph$ is said to \emph{define} the stream $\interpret{\groot} \in \str{\aalph}$.
  (We note that $\sinterpret{}$ is unique by Lemma~\ref{lem:interpret:unique}, below.)

  Let $\astr \in \str{\aalph}$.
  The \emph{canonical $\base$\nb-observation graph of $\astr$} 
  is defined as 
  the sub-coalgebra of the $\safunct$\nb-coalgebra 
  $\pair{\str{\aalph}}{\base}$ generated by $\astr$, that is, 
  the observation graph $\pair{T}{\base}$ with root $\astr$
  where $T \subseteq \str{\aalph}$ 
  is the least set containing $\astr$ that is closed under $\op_1,\ldots,\op_k$.
  The set $\der{\base}{\astr}$ of \emph{$\base$-derivatives of $\astr$}
  is the set of elements of the canonical observation graph of $\sigma$.
\end{definition}

\begin{lemma}\label{lem:interpret:unique}
  For every $\base$-observation graph the mapping $\sinterpret$ is unique whenever it exists. 
\end{lemma}

For the cobasis $\obase{k}$, the existence of $\sinterpret$ is guaranteed:
\begin{proposition}\label{prop:ok=ok}
  The stream coalgebra $\pair{\str{\aalph}}{\obase{k}}$ is final for the functor~$\afunct{X} = \aalph \times X^k$.
  As a consequence, we have that every $\safunct$\nb-coalgebra is an $\obase{k}$\nb-observation graph. 
\end{proposition}

In contrast, the existence of $\sinterpret$ is \emph{not} guaranteed for $\nbase{k}$.
The coalgebra $\pair{\str{\aalph}}{\nbase{k}}$ is final
for a subset of $\safunct$\nb-coalgebras, called \emph{zero-consistent}, see further~\cite{kupk:rutt:2011}.

%
%

\begin{definition}\label{def:bis}
  Let $\base = \pair{\shead}{\tuple{\op_1,\ldots,\op_k}}$ be a stream cobasis.
  A \emph{bisimulation} between $\base$-observation graphs 
  $\aograph = \pair{S}{\pair{o}{n}}$ 
  and $\aograph' = \pair{S'}{\pair{o'}{n'}}$ 
  is a relation $R \subseteq S \times S'$ such that 
  for all $\pair{s}{s'} \in R$ we have that 
  $o(s) = o'(s')$ and $\pair{n_i(s)}{n'_i(s')} \in R$
  for all $1 \le i \le k$, where $n_i$ denotes the $i$-th projection on $n$.
  Two observation graphs are \emph{bisimilar} if there is a bisimulation relating
  their roots. 
\end{definition}

For deterministic transition systems, such as observation graphs, 
bisimilarity coincides with trace equivalence.
As a consequence, the algorithm of Hopcroft--Karp \cite{hopc:karp:1971} is applicable.

\begin{proposition}\label{prop:bisim}
  Bisimilarity of finite $\base$-observation graphs is decidable 
  (in linear time with respect to the sum of the number of vertices).
\end{proposition}

\begin{proposition}\label{prop:graph:bis:equal}
  Let $\base$ be a stream cobasis. Two $\base$-observation graphs define the same stream
  if and only if they are bisimilar.
\end{proposition}

\subsection{From Zip-$k$ Specifications To Observation Graphs}

We construct observation graphs for zip-$k$ specifications.
\begin{definition}\label{def:trsobs}
  Let $\mcl{X} = \{ \rv{X}_0,\ldots,\rv{X}_{n-1}\}$ be a set of recursion variables
  and $\aalph$ a finite set of data-constants.
  Let $k \in \nat$, and $\mcl{S}$ be a zip-$k$ specification over $\pair{\aalph}{\mcl{X}}$.
  We define the orthogonal term rewrite system $\trsobs$ 
  to consist of the following rules:
  \begin{align*}
    \head{a \xcons \astr} & \to a \\
    \pjx{0}{k}{a \xcons \astr} & \to a \xcons \pjx{k-1}{k}{\astr} \\ 
    \pjx{i+1}{k}{a \xcons \astr} & \to \pjx{i}{k}{\astr} && (0 \le i \lt k+1) \\ 
    \head{\zipn{k}{\iastr{0},\ldots,\iastr{k-1}}} & \to \head{\iastr{0}} \\
    \pjx{i}{k}{\zipn{k}{\iastr{0},\ldots,\iastr{k-1}}} & \to \iastr{i} && (0 \le i \lt k)
  \end{align*}
  and additionally for every equation $\Zvar{j} = t$ of $\mcl{S}$ the rules
  \begin{align*}
    \head{\Zvar{j}} & \to \head{t} \\
    \pjx{i}{k}{\Zvar{j}} & \to \pjx{i}{k}{t} && (0 \le i \le k+1)
  \end{align*}
  where the $\Zvar{j}$ are treated as constant symbols.
\end{definition}

Whenever $\aspec$ is clear from the context, 
then by $\nf{t}$ we denote the unique normal form of term $t$ with respect to $\trsobs$.

\newcommand{\synder}[2]{\delta_{#1}(#2)}
\begin{definition}
  Let $\aspec$ be a productive, flat zip-$k$ specification with root~$\Zroot$.
  The set \emph{$\synder{k}{\aspec}$} is the least set containing $\Zroot$~that 
  is closed under $\mylam{t}{(\nf{\pjx{i}{k}{t}})}$ for every $0 \le i < k$.
\end{definition}

\begin{definition}\label{def:zip:graph}
  Let $\aspec$ be a productive, flat zip-$k$ specification with root~$\Zroot$.
  The \emph{$\nbase{k}$\nb-observation graph $\aograph(\aspec)$} is defined as:
  \begin{align*}
    \aograph(\aspec) &= \pair{\synder{k}{\aspec}}{\pair{o}{n}} &
    o(t) &= \nf{\shead(t)} \\ &&
    n(t) &= \tuple{\nf{\pjx{0}{k}{t}},\ldots,\nf{\pjx{k-1}{k}{t}}}
  \end{align*}
  with root $\Zroot$. In words: every node $t$ has
  \begin{enumerate}
    \item the observation $\nf{\head{t}}$ (the label), and
    \item outgoing edges to $\nf{\pjx{0}{k}{t}}$, \ldots, $\nf{\pjx{k-1}{k}{t}}$ (in this order).
  \end{enumerate}
\end{definition}

\begin{lemma}\label{lem:finite:nk}
  Let $\aspec$ be a productive, flat zip-$k$ specification with root~$\Zroot$.
  There exists $m \in \nat$ such that every term in $\synder{k}{\aspec}$ 
  is of the form
  $\cons{d_0}{\cons{\ldots}{\cons{d_{\ell-1}}{\Zvar{j}}}}$ 
  with $\ell \le m$, $d_0,\ldots,d_{\ell-1} \in \aalph$
  and $X_j \in \vars$.
  As a consequence $\synder{k}{\aspec}$ and $\aograph(\aspec)$ are finite.
\end{lemma}

\begin{proof}[Proof Sketch]
  The equations of $\aspec$ are of the form:
  \begin{align*}
    \Zvar{j} = \cons{c_{j,0}}{\cons{\ldots}{\cons{c_{j,m_j-1}}{\zipn{k}{\Zvar{j,0},\ldots,\Zvar{j,k-1}}}}}
    \;\, (0 \le j < n)
  \end{align*}
  Let $m \defdby \max\, \{\,m_i \mid 0 \le i < n \,\}$.
  It suffices that the claimed shape is
  closed under $\mylam{s}{\nf{\pjx{i}{k}{s}}}$ for $0 \le i < k$.~This 
  follows by a straightforward 
  application of Definition~\ref{def:trsobs} 
  together with a precise counting of the `produced' elements.
\end{proof}



We need to ensure that the rewrite system from Definition~\ref{def:trsobs}
implements (is sound for) the intended semantics;
recall that $\aspec$ has a unique solution 
$\sinterpret{}^\aspec : \mcl{X} \to \str{\aalph}$ due to productivity:
\begin{lemma}\label{lem:flat:zipk:ngraph}
  Let $\aspec$ be a productive, flat zip-$k$ specification with root~$\Zroot$.
  For every $t\in \synder{k}{\aspec}$ and $0 \le i < k$ we have~that
  $\head{t} \mred \head{\interpret{t}}$
  and $\interpret{\nf{\pjx{i}{k}{t}}} = \pjx{i}{k}{\interpret{t}}$.
  Hence,
  the graph $\aograph(\aspec)$ is 
  an $\nbase{k}$-observation graph defining $\interpret{\Zroot}^{\aspec}$.
\end{lemma}
\begin{proof}
  The extension of $\sinterpret{}_\alpha$ from Definition~\ref{def:model},
  interpreting the symbols $\spjx{i}{k}$
  by the stream function $\spjx{i}{k} : \str{\aalph} \to \str{\aalph}$ for every $0 \le i < k$,
  is a model of $\trsobs$.
\end{proof}

As an application of Lemmas~\ref{lem:flat} and~\ref{lem:flat:zipk:ngraph}
we get
\begin{lemma}\label{lem:zipk:ngraph}
  For every productive zip-$k$ specification with root~$\Zroot$
  we can construct an $\nbase{k}$\nb-observation graph 
  defining the stream $\interpret{\Zroot}^{\aspec}$.
\end{lemma}

We arrive at our first main result:
\begin{theorem}\label{thm:zipk:decide}
  Equivalence of zip-$k$ specifications is decidable.
\end{theorem}
\begin{proof}
  Lemma~\ref{lem:solve} allows to reduce the equivalence problem for
  unproductive zip\nb-$k$ specifications to 
  a finite number of equivalence problems for productive zip\nb-$k$ specifications.
  Propositions~\ref{prop:graph:bis:equal}, \ref{prop:bisim} and Lemma~\ref{lem:zipk:ngraph} 
  imply decidability of equivalence for productive zip-$k$ specifications.
\end{proof}

\begin{proposition}\label{prop:flat:ptime}
  Equivalence of productive, flat \zspec{s} is decidable in quadratic time.
\end{proposition}

\begin{example}\label{ex:morse:bis}
  Consider the zip-$2$ specification with root $\con{N}$:
  \begin{align*}
      \con{N} &= \cons{0}{\zip{\cons{1}{\con{W}}}{\cons{1}{\con{U}}}} &
      \con{U} &= \cons{1}{\zip{\con{V}}{\con{U}}} \\
      \con{V} &= \cons{0}{\zip{\con{V}}{\cons{1}{\con{U}}}} &
      \con{W} &= \zip{\con{N}}{\con{V}}
  \end{align*}
    \begin{center}
    \begin{small}
        \begin{tikzpicture}[default]

    \node [state] (q0) {$\con{M}$}; 
    \node (h0) [lhead=q0]{$0$};
    \node [state] (q1) [bl=q0] {$0 \xcons \con{X}$}; 
    \node (h1) [rhead=q1]{$0$};
    \node [state] (q2) [br=q0] {$1 \xcons \con{Y}$}; 
    \node (h2) [lhead=q2] {$1$};
    \path (q0) 
               edge node [lal] {$\seven$} (q1)
               edge node [lar] {$\sodd$} (q2)
          (q1) 
               edge [bloop] node {$\seven$} (q1)
               edge [bend left=25] node [below] {$\sodd$} (q2)
          (q2) 
               edge [bloop] node {$\seven$} (q2)
               edge [bend left=25] node [above] {$\sodd$} (q1);

    \node [state,right of=q0,node distance=45mm] (u) {$\con{N}$};
    \node (u0) [rhead=u,yshift=1mm] {$0$};
    \node [state,bl=u] (l) {$0:1:\con{U}$}; 
    \node (l0) [rhead=l] {$0$};
    \node [state,br=u] (r) {$1:\con{W}$}; 
    \node (r0) [lhead=r] {$1$};
    \node [state,br=l] (b) {$1:\con{V}$}; 
    \node (b0) [rhead=b] {$1$};

    \path 
      (u) edge node [lal,pos=.2] {$\seven$} (l)
      (u) edge [bend left] node [lar,pos=.7]{$\sodd$} (r) 
      (l) edge [tloop] node [left,pos=.8] {$\seven$} (l) 
      (l) edge [bend left] node [lar] {$\sodd$} (b) 
      (r) edge node [lbr,pos=.3] {$\seven$} (b) 
      (r) edge [bend left] node [lar] {$\sodd$} (u) 
      (b) edge [out=-20,in=20,looseness=15] node [left] {$\seven$} (b) 
      (b) edge [bend left] node [lar] {$\sodd$} (l);
      
    \path [dashed,-]
      (q0) edge [bend left=12] (u)
      (q1) edge [out=30,in=175] (u)
      (q1) edge [bend right=50] (l)
      (q2) edge [out=55,in=165] (r)
      (q2) edge [bend right=30] (b);

  \end{tikzpicture}

    \end{small}
    \end{center}
    \vspace{-6ex}
  Its $\nbase{2}$-observation graph is depicted on the right above.
  The dashed lines indicate a bisimulation with the observation graph 
  from Fig.~\ref{fig:ograph:morse} here depicted on the left.
\end{example}

\subsection{From Observation Graphs To Zip-$k$ Specifications}

%

\begin{lemma}\label{lem:obase:zip}
  The canonical $\obase{k}$-observation graph of a stream $\astr \in \str{\aalph}$
  is finite
  if and only if
  $\astr$ can be defined by a zip-$k$ specification consisting of equations of the form:
  \begin{align*}
    \Zvar{i} = a_i \xcons \zipn{k}{\Zvar{i,1},\Zvar{i,2},\dots,\Zvar{i,k}}
  \end{align*}
\end{lemma}
\begin{proof}
  For the translation forth and back, 
  it suffices to observe the correspondence 
  between an equation $\rv{Y} = a \xcons \zipn{k}{\rv{Y}_1,\ldots,\rv{Y}_k}$
  and its semantics $\head{\interpret{\rv{Y}}} = a$, 
  $\pjx{1}{k}{\interpret{\rv{Y}}} = \interpret{\rv{Y}_1}$, 
  \dots, 
  $\pjx{k}{k}{\interpret{\rv{Y}}} = \interpret{\rv{Y}_k}$.
\end{proof}

\begin{lemma}\label{lem:nbase:zip}
  The canonical $\nbase{k}$-observation graph of a stream $\astr \in \str{\aalph}$ is finite
  if and only if
  $\astr$ can be defined by a zip-$k$ specification consisting of pairs of equations of the form:
  \begin{align*}
    \Zvar{i} & = a_i \xcons \rv{X}'_i \\
    \rv{X}'_i & = \zipn{k}{\Zvar{f(i,1)},\Zvar{f(i,2)},\dots,\Zvar{f(i,k-1)},\rv{X}'_{f(i,0)}}
  \end{align*}
  over recursion variables $\mcl{X} \cup \mcl{X}'$
  where $\mcl{X} = \{ \Zvar{0},\ldots,\Zvar{n-1}\}$ 
  and $\mcl{X}' = \{\rv{X}'_i \where \Zvar{i} \in \mcl{X} \}$,
  and $f : \fin{n} \times \fin{k} \to \fin{n}$ 
  such that $a_{f(i,0)} = a_i$ for all $i\in\fin{n}$.
\end{lemma}

\begin{proof}
  If $\rv{Y} = a : \rv{Y}'$ and $\rv{Y}' = \zipn{k}{\rv{Y}_1,\dots,\rv{Y}_{k-1},\rv{Y}'_0}$
  then $\head{\interpret{\rv{Y}}} = a$, 
  $\pjx{0}{k}{\interpret{\rv{Y}}} = a : \interpret{\rv{Y}'_0}$,
  and 
  $\pjx{i}{k}{\interpret{\rv{Y}}} = \interpret{\rv{Y}_i}$ ($1 \le i \lt k$).
  Since there also is an equation $\rv{Y}_0 = a : \rv{Y}'_0$,
  it holds that 
  $\interpret{\rv{Y}'_0} = \tail{\interpret{\rv{Y}_0}}$
  and hence $\pjx{0}{k}{\interpret{\rv{Y}}} = \interpret{\rv{Y}_0}$.
\end{proof}

\section{Automaticity and Observation Graphs}\label{sec:k-auto}
After our first main result (Theorem~\ref{thm:zipk:decide}) 
we proceed with connecting zip-$k$ specifications to $k$\nb-automatic sequences.

\subsection{Automatic Sequences}

\begin{definition}[\cite{allo:shal:2003}]\label{def:DFAO}
  A \emph{deterministic finite automaton with output (DFAO)} 
  is a tuple $\tuple{Q,\Sigma,\delta,q_0,\aalph,\lambda}$
  where 
  \begin{itemize}
    \item 
      $Q$ is a finite set of states,
    \item 
      $\Sigma$ a finite input alphabet, 
    \item 
      $\delta : Q \times \Sigma \to Q$ a transition function,
    \item 
      $q_0 \in Q$ the initial state,
    \item 
      $\aalph$ an output alphabet, and
    \item 
      $\lambda : Q \to \aalph$ an output function.
  \end{itemize}
  We extend $\delta$ to words over $\Sigma$ as follows:
  \begin{align*}
    \delta(q,\emptyword) &= q && \text{for $q \in Q$} \\
    \delta(q,wa) &= \delta(\delta(q,a),w) && \text{for $q \in Q, a \in \Sigma, w \in \Sigma^*$}
  \end{align*}
  and we write $\delta(w)$ as shorthand for $\delta(q_0,w)$.
\end{definition}

For $n,k\in \nat$, $k \ge 2$,
we use $(n)_k$ to denote the representation of $n$ with respect to the base $k$
(without leading zeros).
More precisely, for $n > 0$ we have $(n)_k = n_m n_{m-1} \ldots n_0$ where $0 \le n_m, \ldots, n_0 < k$,
$n_m > 0$ and $n = \sum_{i=0}^{m} n_i k^i$;
for $n=0$ we fix $(n)_k = \emptyword$.

\begin{definition}
  A $k$-DFAO $A$ is a DFAO~$\tuple{Q,\Sigma,\delta,q_0,\aalph,\lambda}$ 
  with input alphabet $\Sigma = \nat_{<k}$.
  For $q \in Q$, we define a stream $\gen{A}{q}$ by:
  $\gen{A}{q}(n) = \lambda(\delta(q,(n)_k))$ for every $n \in \nat$.

  We write $\genz{A}$ as shorthand for $\gen{A}{q_0}$.
  Moreover, we say that the automaton $A$ \emph{generates} the stream $\genz{A}$.
\end{definition}

\begin{definition}\label{def:auto}
  A stream $\astr : \str{\aalph}$
  is called \emph{$k$\nb-automatic}
  if there exists a $k$-DFAO that generates $\astr$.
  A stream is called \emph{automatic} if it is $k$-automatic for some $k\ge2$.
\end{definition}

The exclusion of leading zeros in 
the number representation $(n)_k$ 
is not crucial for the definition of automatic sequences.
Every $k$\nb-DFAO  can be transformed into an equivalent $k$\nb-DFAO that
ignores leading zeros:

\begin{definition}
  A $k$-DFAO $\tuple{Q,\Sigma,\delta,q_0,\aalph,\lambda}$ is called \emph{invariant under leading zeros}
  if for all $q \in Q$: $\lambda(q) = \lambda(\delta(q,0))$.
\end{definition}

\newcommand{\mycite}{\cite[Theorem 5.2.1 with Corollary 4.3.4]{allo:shal:2003}}
\begin{lemma}[\mycite]\label{lem:dfao:zeros}
  For every $k$-DFAO $\aaut$ there is a $k$-DFAO $\aaut'$
  that is invariant~under leading zeros and
  generates the same stream ($\genz{\aaut} = \genz{\aaut'}$).
\end{lemma}


Automatic sequences can be characterized in terms of their `kernels' being finite.  
Kernels of a stream $\astr$ are sets of arithmetic subsequences of $\astr$, defined as follows.
\begin{definition}\label{def:kernel}
  The \emph{$k$-kernel} of a stream $\astr \in \str{\aalph}$
  is the set of subsequences $\{\pjx{i}{k^p}{\astr} \where p \in \nat, i < k^p\}$.
\end{definition}

{
\renewcommand{\mycite}{\cite[Theorem~6.6.2]{allo:shal:2003}}
\begin{lemma}[\mycite]\label{lem:kernel}
  A stream $\astr$ is $k$-automatic if and only if the
  $k$-kernel of $\astr$ is finite. 
\end{lemma}
}

\subsection{Observation Graphs and Automatic Sequences}

There is a close correspondence between observation graphs 
with respect to the cobasis $\nbase{k}$ and $k$-DFAOs.
For $k$-DFAOs $A$ that are invariant under leading zeros
an edge $q \to p$ labeled $i$
implies that the stream generated by $p$ is the $\spjx{i}{k}$-projection
of the stream generated by $q$, that is, $\gen{A}{p} = \pjx{i}{k}{\gen{A}{q}}$.
The following lemma treats the case of general $k$\nb-DFAOs.
 
\begin{lemma}\label{lem:dfao:pjx}
  Let $\aaut = \tuple{Q,\Sigma,\delta,q_0,\aalph,\lambda}$ be a $k$-DFAO.
  Then for every $q \in Q$ we have:
  $\tail{\gen{A}{\delta(q,0)}} = \tail{\pjx{0}{k}{\gen{A}{q}}}$
  and for all $1 \le i < k$:
  \begin{align}
    \gen{A}{\delta(q,i)} = \pjx{i}{k}{\gen{A}{q}} \label{dfao:pjx}
  \end{align}
  Hence, if $\aaut$ is invariant under leading zeros,
  then property~\eqref{dfao:pjx} holds for all $0 \le i < k$.
\end{lemma}

\begin{proof}
  Follows immediately from
  $(kn + i)_k = (n)_k i$
  for all $n \in \nat$ and $0 \le i < k$ such that $n \ne 0$ or $i \ne 0$.
\end{proof}

As a consequence of Lemma~\ref{lem:dfao:pjx} we have that $k$-DFAOs, that are invariant under leading zeros, 
are $\nbase{k}$-observation graphs
for the streams they define, and vice versa.
Formally, this is just a simple change of notation%
\footnote{Note that even this small change of notation can be avoided by introducing $k$-DFAOs as
          coalgebras over the functor $\afunct{X} = \aalph \times X^k$ as well.}:
\begin{definition}
  Let $\aaut = \tuple{Q,\Sigma,\delta,q_0,\aalph,\lambda}$
  be a $k$-DFAO that is invariant under leading zeros.
  We define
  the $\nbase{k}$-observation graph $\dfaoograph{\aaut} = \pair{Q}{\pair{o}{n}}$
  with root $q_0$
  where for every $q \in Q$:
  $o(q) = \lambda(q)$,
  $n_i(q) = \delta(q,i)$ for $i < k$, and
  $\interpret{q} = \gen{A}{q}$.

  Let $\aograph = \pair{S}{\pair{o}{n}}$ be an $\nbase{k}$-observation graph 
  over $\aalph$ with root $r\in S$.
  Then we define a $k$-DFAO $\ographdfao{\aograph}$
  as follows: $\ographdfao{\aograph} = \tuple{Q,\nat_{<k},\delta,q_0,\aalph,\lambda}$
  where $Q = S$, $q_0 = r$, and for every $s \in S$:
  $\lambda(s) = o(s)$, and $\delta(s,i) = n_i(s)$ for $i < k$.
\end{definition}

\begin{proposition}
  For every $k$-DFAO $\aaut$ that is invariant under leading zeros,
  the $\nbase{k}$-observation graph $\dfaoograph{\aaut}$
  defines the stream that is generated by $\aaut$.
  
  Conversely, we have for every $\nbase{k}$-observation graph~$\aograph$,
  that the $k$-DFAO $\ographdfao{\aograph}$ is invariant under zeros and
  generates the stream defined by $\aograph$.
\end{proposition}

Another way to see the correspondence between automatic sequences
and their finite, canonical $\nbase{k}$-observation graphs is as follows.
The elements of the canonical observation graph of a stream $\astr$, that is, 
the set of $\{\spjx{0}{k},\ldots,\spjx{k-1}{k}\}$-derivatives of $\astr$,
coincide with the elements of the $k$-kernel of $\astr$.
This is used in the proof of the following theorem.

\begin{proposition}\label{thm:auto:nk}
  For streams $\astr \in \str{\aalph}$ the following properties are equivalent:
  \begin{enumerate}
    \item The stream $\astr$ is $k$-automatic.
    \item The canonical $\nbase{k}$-observation graph of $\astr$ is finite. 
  \end{enumerate}
\end{proposition}

\begin{proof}
  The equivalence of (i) and (ii) is a consequence of Lemma~\ref{lem:kernel}
  in combination with the observation that the set of functions
  $\{\spjx{i}{k^p} \where p \in \nat, i < k^p\}$
  coincides with the set of functions obtained from 
  arbitrary iterations of functions $\spjx{0}{k}$, \ldots, $\spjx{k-1}{k}$
  (that is, function compositions $\gamma_1 \cdot \ldots \cdot \gamma_n$ with $n\in\nat$ and $\gamma_i \in \{\spjx{0}{k},\ldots,\spjx{k-1}{k}\}$).
\end{proof}

Proposition~\ref{thm:auto:nk} gives a coalgebraic perspective on automatic sequences.
Moreover, it frequently allows for simpler proofs or disproofs
of automaticity than existing characterizations.
For example, in the following sections we will
derive observation graphs for streams that are specified by zip-specifications.
Then it is easier to stepwise iterate the finite
set of functions $\{\spjx{0}{k},\ldots,\spjx{k-1}{k}\}$
than to reason about infinitely many subsequences in the kernel
$\{\pjx{i}{k^p}{\astr} \where p \in \nat, i < k^p\}$.

Proposition~\ref{thm:auto:nk} was independently found by Kupke and Rutten, 
see Theorem~8 in their recent report~\cite{kupk:rutt:2011}.

We arrive at our second main result:
\begin{theorem}\label{thm:zip:auto}
  For streams $\astr \in \str{\aalph}$ the following properties are equivalent:
  \begin{enumerate}
    \item\label{i:auto} The stream $\astr$ is $k$-automatic.
    \item\label{i:zip} The stream $\astr$ can be defined by a zip-$k$ specification.
    \item\label{i:ngraph} The canonical $\nbase{k}$-observation graph of $\astr$ is finite. 
    \item\label{i:ograph} The canonical $\obase{k}$-observation graph of $\astr$ is finite. 
  \end{enumerate}
\end{theorem}

\begin{proof}
  We have that $\ref{i:auto} \Leftrightarrow \ref{i:ngraph}$ by Theorem~\ref{thm:auto:nk},
  $\ref{i:ngraph} \Rightarrow \ref{i:zip}$ by Lemma~\ref{lem:nbase:zip},
  and 
  $\ref{i:zip} \Rightarrow \ref{i:ngraph}$ by Lemma~\ref{lem:zipk:ngraph}.
  Moreover,
  it holds that $\ref{i:ograph} \Rightarrow \ref{i:zip}$ by Lemma~\ref{lem:obase:zip}.

  Finally, we show $\ref{i:ngraph} \Rightarrow \ref{i:ograph}$.
  Assume $\aograph = \pair{S}{\pair{o}{n}}$ 
  is a finite $\nbase{k}$\nb-observa{-}tion graph with root $r$ defining $\astr$
  and let $\sinterpret_{\aograph} \funin S \to \str{\aalph}$ be the unique
  $F$\nb-homomorphism into $\pair{\str{\aalph}}{\nbase{k}}$.
  Let $n = \tuple{n_1,\dots,n_k}$.
  Then $o(s) = \head{\interpret{s}_{\aograph}}$ 
  and $\interpret{n_i(s)}_{\aograph} = \pjx{i-1}{k}{\interpret{s}_{\aograph}}$ for all $1 \le i \le k$ and $s \in S$.
  We define $\aograph' = \pair{S'}{\pair{o'}{n'}}$ 
  where $S' = S \cup \{ \utail{s} \where s \in S \}$, 
  $o'(s) = o(s)$, 
  $o'(\utail{s}) = o(n_2(s))$
  $n'_i(s) = n_{i+1}(s)$ for $1 \le i \lt k$,
  $n'_k(s) = \utail{n_1(s)}$,
  $n'_i(\utail{s}) = n_{i+2}(s)$ for $1 \le i \le k-2$,
  $n'_{k-1}(\utail{s}) = \utail{n_1(s)}$ and
  $n'_{k}(\utail{s}) = \utail{n_2(s)}$
  with root $r \in S'$.
  Let $\sinterpret_{\aograph'} \funin S' \to \str{\aalph}$ be defined by
  $\interpret{s}_{\aograph'} = \interpret{s}_{\aograph}$ and
  $\interpret{\utail{s}}_{\aograph'} = \tail{\interpret{s}_{\aograph}}$.
  It can be checked that $\sinterpret_{\aograph'}$
  is an $\safunct$-homomorphism into $\pair{\str{\aalph}}{\obase{k}}$
  with $\astr = \interpret{r}_{\aograph'}$.
  Hence $\aograph'$ is an $\obase{k}$\nb-observation graph defining~$\astr$.
%
\end{proof}


\section{A Dynamic Logic Representation \\ of Automatic Sequences}\label{sec:dynamic}
This section connects automatic sequences with expressivity in a \emph{propositional dynamic logic} (PDL)
derived from the cobases $\nbase{k}$ and $\obase{k}$.
For simplicity, we shall restrict attention to the case of
$\Delta = \set{0,1}$ and $\nbase{2} = \tuple{\shead,\seven,\sodd}$. 

The set of \emph{sentences} $\phi$ and \emph{programs} $\pi$ of 
our version of PDL is given by the following BNF grammar:
\begin{align*}
  \phi \BNFis 0 \BNFor 1 \BNFor \nott\phi \BNFor \phi\andd\phi \BNFor [\pi]\phi\\
  \pi \BNFis \seven \BNFor \sodd \BNFor \pi;\pi \BNFor \pi \sqcup \pi \BNFor \pi^{\ast}
\end{align*}
We interpret PDL in an arbitrary $F$-coalgebra $\GG = \tuple{S, \tuple{o,n}}$.
Actually, we can be more liberal and interpret PDL in \emph{models} of the form
\[ 
  \GG = \tuple{S, 0, 1, \seven,\sodd}
\]
where $0 \subseteq S$, $1\subseteq S$, $\seven\subseteq S^2$, and $\sodd\subseteq S^2$.
These are more general than $F$-coalgebras because we do not insist that
$0\cap 1 = \emptyset$, or that $\seven$ and $\sodd$ be interpreted as functions.
Nevertheless, these extra properties do hold in the intended model
\[
  \tuple{\Delta^\omega, 0,1,\seven,\sodd}
\]
where $0$ is the set of streams whose head is the number $0$,
and similarly for $1$; $(\sigma,\tau)\in \seven$ iff $\tau = (\sigma_0,\sigma_2,\sigma_4,\ldots)$,
and similarly for $\sodd$.

The interpretation of each sentence $\phi$ is a subset of $S$;
 the interpretation of each program $\pi$ is a relation on $S$, that is, a subset of $S\times S$.
The definition is as usual for PDL:
\begin{gather*}
\begin{aligned}
  \semantics{0} & = \set{x\in S: x\in 0}
  & \semantics{\seven} & = \seven \\
  \semantics{1} & = \set{x\in S: x\in 1}
  & \semantics{\sodd} & = \sodd \\
  \semantics{\phi\andd\psi} & = \semantics{\phi}\cap\semantics{\psi} 
  & \semantics{\pi_1;\pi_2} & = \semantics{\pi_1};\semantics{\pi_2} \\
  \semantics{\nott\phi} & = S\setminus \semantics{\phi} 
  & \semantics{\pi_1\sqcup \pi_2} & = \semantics{\pi_1}\cup \semantics{\pi_2} \\
  &
  & \semantics{\pi^*} & = \semantics{\pi}^*
\end{aligned} \\
  \semantics{[\pi]\phi} = \set{x : (\forall y)(\pair{x}{y}\in \semantics{\pi} \rightarrow y\in\semantics{\phi})}
\end{gather*}
In words, we interpret $\seven$ and $\sodd$ by themselves
that correspond in the given model.
We interpret $;$ by relational composition,
$\sqcup$ by union of relations, 
${}^*$ by Kleene star (= reflexive-transitive closure) 
of relations, and we use the usual boolean operations and
dynamic modality $[\pi]\phi$.

We use the standard boolean abbreviations for $\phi\iif\psi$ and $\phi\iiff\psi$,
and of course we use the standard semantics.
We also write $\dmd{\pi}\phi$ for $\nott[\pi]\nott\phi$;
again this is standard.

For example, let $\chi$ be the sentence $[(\seven\sqcup\sodd)^*](0\iiff \nott 1)$.
Then in any model $\GG$, a point $x$ has $x\models\chi$ iff
for all points $y$ reachable from $x$ in zero or more steps in the relation
$\seven\cup\sodd$, $y$ satisfies exactly one of $0$ or $1$.

\begin{proposition}\label{prop-preservation}
  If $f: M\to N$ is a morphism of models 
  and $x\models \phi$ in $M$, then $f(x)\models\phi$ in $N$.
\end{proposition} 

\begin{proposition}
  For every finite pointed model $\tuple{\GG, x}$ there is a sentence
  $\phi_x$ of PDL so that for all (finite or infinite) $F$-coalgebras  $\tuple{\HH, y}$,
  the following are equivalent:
  \begin{enumerate}
    \item  $y\models \phi_x$ in $\HH$.
    \item There is a bisimulation between $\GG$ and $\HH$ which relates $x$ to $y$.
  \end{enumerate}
  We call $\phi_x$ the \emph{characterizing sentence of $x$}.
  \label{prop-characterization}
\end{proposition}

For infinitary modal logic, this result was shown in~\cite{barw:moss:1996}, 
and the result here for PDL is a refinement of it.

For example, we construct a characterizing sentence for the Thue--Morse sequence~$\con{M}$,
see Fig.~\ref{fig:ograph:morse}. 
Let $\phi$ and $\psi$ be given~by 
\begin{align*}
  \phi & = 0 \andd \nott 1 \andd  \dmd{\seven}0\andd\bx{\seven}0 \andd \dmd{\sodd}1 \andd \bx{\sodd}1 \\
  \psi & = \nott 0 \andd  1 \andd \dmd{\seven}1\andd\bx{\seven}1 \andd \dmd{\sodd}0 \andd \bx{\sodd}0
\end{align*}
Then $\phi_{\con{M}} = \phi \andd [(\seven\sqcup\sodd)^*](\phi \vee \psi)$
is a characteristic sentence of the top node in Fig.~\ref{fig:ograph:morse};
$\phi_{\con{M}}$ also characterizes $\con{M}$~in the following sense:
the only stream $\sigma$ such that $\sigma \models \phi_{\con{M}}$ is~$\con{M}$.

\begin{proposition}
  \label{prop-KP-deterministic}
  The following finite model properties hold:
  \begin{enumerate}
    \item 
      If a sentence~$\phi$ 
      has a model, it has a finite model~\cite{koze:pari:1981}.
    \item 
      If $\phi$ has a model in which $\seven$ and $\sodd$
      are total functions, then it has a finite model with these properties~\cite{ben-:halp:pnue:1982}.
  \end{enumerate}
\end{proposition}

\begin{remark}
  Our statement of the second result is a slight variation of what appears in~\cite{ben-:halp:pnue:1982}.
\end{remark}

We arrive at our third main result:
\begin{theorem}
  The following are equivalent for $\sigma\in \Delta^\omega$:
  \begin{enumerate}
    \item 
      $\sigma$ is $2$-automatic.
    \item  
      There is a sentence $\phi$ such that for all $\tau\in \Delta^\omega$,
      $\tau\models\phi$ in $\tuple{\Delta^\omega, \tuple{\shead,\seven,\sodd}}$ 
      iff $\tau = \sigma$.
  \end{enumerate}
\end{theorem}

\begin{proof}
  $(i) \Rightarrow (ii)$:
  Let $\sigma$ be automatic,
  and let $M$ be a finite $F$-coalgebra and $x\in M$ be such that 
  the unique coalgebra morphism $f : M \to \Delta^\omega$
  has $f(x) = \sigma$.
  Let $\phi_x$ be the characterizing sentence of $x$ in $M$,
  using Proposition~\ref{prop-characterization}.  
  By Proposition~\ref{prop-preservation}, 
   $\sigma\models\phi_x$ in $\Delta^\omega$.
  Now suppose that $\tau\models\phi_x$ in $\Delta^\omega$.
  Since $\phi_x$ is a characterizing sentence, there is a bisimulation on $\Delta^\omega$
  relating $\sigma$ to $\tau$.  
  By Lemma~\ref{lem:coinduction},
  $\sigma = \tau$.

  $(ii) \Rightarrow (i)$:
  Let $\phi$ be a sentence with the property that $\sigma$ is the only
  stream which satisfies $\phi$.    Since $\sigma$ has a model,
  it has a finite model, by Proposition~\ref{prop-KP-deterministic}.   
  Moreover, this model $M$ may be taken to be a finite  $F$-coalgebra
  with a distinguished point $x$.
  By \cite[Theorem~5]{kupk:rutt:2011} 
  let $\phi : M \to \Delta^\omega$ be the unique coalgebra morphism.
  Let $\tau = \phi(x)$. 
  Since $M$ is finite, $\tau$ is automatic. 
  By Proposition~\ref{prop-characterization}, $\tau\models\phi$ in $\Delta^\omega$.
  But by the uniqueness assertion in part (2) of our theorem, we must have
  $\tau = \sigma$.
  Therefore $\sigma$ is automatic.
\end{proof}

\section{Mix-Automaticity}\label{sec:mix-auto}
The zip-specifications considered so far were uniform, 
all zip-operations in a zip-$k$ specification have the same arity~$k$.
Now we admit different arities of zip in one zip-specification (Definition~\ref{def:Zterms}).
To emphasize the difference with zip\nb-$k$ specifications 
we will here speak of \emph{zip\nb-mix} specifications.
This extension leads to a proper extension of automatic sequences 
and some delicate decidability problems.

\begin{definition}\label{def:sda-DFAO}
  A \emph{state-dependent-alphabet DFAO} 
  is a tuple $\tuple{Q,\Sigma,\delta,q_0,\aalph,\lambda}$, 
  where 
  \begin{itemize}
    \item 
      $Q$ is a finite set of states,
    \item 
      $\Sigma = \{\Sigma_q\}_{q \in Q}$ a family of input alphabets, 
    \item 
      $\delta = \{ \delta_q : \Sigma_q \to Q \}_{q \in Q}$ a family of transition functions,
    \item 
      $q_0 \in Q$ the initial state,
    \item 
      $\aalph$ an output alphabet, and
    \item 
      $\lambda : Q \to \aalph$ an output function.
  \end{itemize}
  \noindent
  We write $\delta(q,i)$ for $\delta_q(i)$ iff $i\in\Sigma_q$,
  and extend $\delta$ to words as follows:
  Let $q\in Q$ and $w = a_{n-1} \ldots a_0$
  where $a_i \in \Sigma_{r_i}$ ($0 \le i \lt n$)
  with $r_i\in Q$ defined by:
  $r_0 = q$ and $r_{i+1} = \delta(r_i,a_i)$.
  Then we let $\delta(q,w) = r_n$.
\end{definition}

A {state-dependent-alphabet DFAO} can be seen as a DFAO
whose transition function is a \emph{partial} map
$\delta : Q \times \bigcup \Sigma \pto Q$ 
such that $\delta(q,a)$ is defined iff $a\in\Sigma_q$.

We use this concept to generalize $k$-DFAOs
where the input format are numbers in base $k$
by the following two\nb-tiered construction.
We define $P$\nb-DFAOs where $P$ is
a DFAO determining the base of each digit 
depending on the digits read before.
Thus $P$ can be seen as fixing a variadic numeration system.
For example, for ordinary base $k$ numbers,
we define $P$ to consist of a single state $q$ with output $k$
and edges $0,\ldots,k-1$ looping to itself.

\begin{definition}\label{def:sdb-DFAO}
  A \emph{base determiner~$\abs$} is a state-dependent-alphabet DFAO 
  of the form
  \( \abs = \tuple{Q,\{\fin{\beta(q)}\}_q,\delta,q_0,\nat,\beta} \).
  The \emph{base\nb-$\abs$ representation of $n \in \nat$} 
  is defined by 
  \begin{align*}
    (n)_{\abs} = (n)_{\abs,q_0}
    && \text{where}
    && (n)_{\abs,q} = (n')_{\abs,\delta(q,d)} \cdot d
  \end{align*}
  with $n' = \floor{\frac{n}{\beta(q)}}$ and $d = [n]_{\beta(q)}$,
  the quotient and the remainder of division of $n$ by $\beta(q)$, 
  respectively.

  A \emph{$\abs$-DFAO~$\aaut$}
  is a state-dependent-alphabet DFAO
  \[ \aaut = \tuple{Q',\{\fin{\beta'(q')}\}_{q' \in Q'},\delta',q'_0,\aalph,\lambda} \]
  \emph{compatible with} $\abs$, i.e.\
  $\tuple{Q',\{\fin{\beta'(q')}\}_{q' \in Q'},\delta',q'_0,\nat,\beta'}$ and $\abs$ 
  are bisimilar.
  
  A \emph{mix-DFAO} is a $\abs$-DFAO for some base determiner $\abs$.
\end{definition}

Note that the output alphabet of a base determiner can be taken to be
finite as the range of $\beta$.
The compatibility of $A$ with $P$ entails that $A$ reads the number format defined by $P$.
Moreover, every mix-DFAO~%
$\aaut = \tuple{Q,\{\fin{\beta(q)}\}_{q \in Q},\delta,q_0,\aalph,\lambda}$
is a $\abs_{\!\aaut}$\nb-DFAO where 
$\abs_{\!\aaut} = \tuple{Q,\{\fin{\beta(q)}\}_{q \in Q},\delta,q_0,\nat,\beta}$.

These DFAOs introduce a new class of sequences, 
which we call `mix-automatic' in order to emphasize the connection 
with zip-mix specifications. 

\begin{definition}\label{def:mix-automatic}
  Let $\abs$ be a base determiner, and 
  $\aaut = \tuple{Q,\{\fin{\beta(q)}\}_{q\in Q},\delta,q_0,\aalph,\lambda}$ a $\abs$\nb-DFAO.
  For states $q \in Q$, we define $\gen{A}{q} \in \str{\aalph}$ by:
  $\gen{A}{q}(n) = \lambda(\delta(q,(n)_{\abs}))$ for all $n \in \nat$.
  We define $\genz{A} = \gen{A}{q_0}$,
  and say $A$ \emph{generates} the stream $\genz{A}$.

  A sequence $\astr\in\str{\aalph}$ is \emph{$\abs$\nb-automatic} 
  if there is a $\abs$\nb-DFAO~$\aaut$ such that $\astr = \genz{\aaut}$.
  A stream is called \emph{mix\nb-automatic} 
  if it is $\abs$\nb-automatic for some base determiner~$\abs$.
\end{definition}
 

\begin{example}\label{ex:sdb:dfao}
  Consider the following mix-DFAO~$\aaut$:
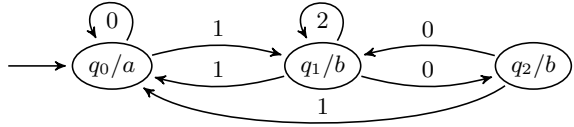
  \begin{center}
    \vspace{-1ex}
    \begin{small}
    \begin{tikzpicture}[->,>=stealth',shorten >=1pt,auto,node distance=2.8cm,semithick]
      \node (s) {};
      \node[state] (q0) [right of=s,node distance=15mm] {$q_0/a$};
      \node[state] (q1) [right of=q0] {$q_1/b$};
      \node[state] (q2) [right of=q1] {$q_2/b$};
      \path (s)  edge (q0)
            (q0) edge [tloop] node {$0$} (q0)
                 edge [bend left=15] node [above] {$1$} (q1)
            (q1) edge [tloop] node {$2$} (q1)
                 edge [bend left=15] node [above] {$1$} (q0)
                 edge [bend right=15] node [above] {$0$} (q2)
            (q2) edge [bend right=15] node [above] {$0$} (q1)
                 edge [out=-145,in=-35,looseness=.6] node [above] {$1$} (q0);
  \end{tikzpicture}
  \end{small}
    \vspace{-1ex}
  \end{center}
  
  \noindent
  We note that $\aaut$ is a $\abs$-DFAO where $\abs$ is the base determiner
  obtained from $\aaut$ by redefining the output for $q_0$, $q_1$ and $q_2$
  as the number of their outgoing edges $2$, $3$ and $2$, respectively.
  
  As an example, we compute $(5)_{\aaut}$, 
  and $(23)_{\aaut}$ as follows:
  \vspace{-.5ex}
  \begin{align*}
    (5)_{q_0}  & = (2)_{q_1} 1 = (0)_{q_2} 2 1 = 2 1 \\
    (23)_{q_0} & = (11)_{q_1} 1 = (3)_{q_1} 2 1 = (1)_{q_2} 0 2 1 = (0)_{q_0} 1 0 2 1 = 1 0 2 1
  \end{align*}
  \vspace{-.5ex}
  where $(n)_{q}$ denotes $(n)_{\aaut,q}$.
  The sequence~$\genz{\aaut}$ 
  begins with
  \[
    a \xcons b \xcons b \xcons a 
      \xcons b \xcons b \xcons a 
      \xcons a \xcons b \xcons b 
      \xcons b \xcons a \xcons a 
      \xcons a \xcons \underline{a} \xcons b 
      \xcons b \xcons b \xcons b 
      \xcons b \xcons b \xcons a 
      \xcons a \xcons \underline{a} \xcons a 
      \xcons b \xcons a \xcons b 
      \xcons \ldots
  \]
  with entries $5$ and $23$ underlined.
  E.g.\ $\lambda(\delta(q_0,1021)) = a$ since starting from $q_0$ and reading $1021$ 
  from right to left brings you back at state $q_0$ with output~$a$.
\end{example}

We briefly indicate how to see that mix-automaticity properly extends automaticity.
Let $\astr$ and $\bstr$ be $k$- and $\ell$-automatic sequences.
Then the stream $\zip{\astr}{\bstr}$ is mix\nb-automatic, 
but not necessarily automatic, 
by Cobham's Theorem~\cite{cobh:1969}.

\begin{proposition}\label{prop:mix:extends}
  The class of mix-automatic sequences extends that of automatic sequences.
\end{proposition}

\begin{definition}\label{def:mix:ograph}
  Let $\kappa : \str{\aalph} \to \nat_{>1}$,
  and let $\sbfunct$ be the functor $\bfunct{X} = \sum_{k=2}^{\infty} \aalph \times X^k$.
  We define the cobasis 
  \[
    \nbase{\kappa} = \pair{\shead}{\mylam{\astr}{\tuple{\pjx{0}{\kappa(\astr)}{\astr},\ldots,\pjx{\kappa(\astr)-1}{\kappa(\astr)}{\astr}}}}
  \]
  An \emph{$\nbase{\kappa}$-observation graph} is 
  a $\sbfunct$-coalgebra $\aograph = \pair{S}{\pair{o}{n}}$ 
  with a distinguished \emph{root} element $\groot \in S$,
  such that there exists a $\sbfunct$\nb-homomorphism $\sinterpret \funin S \to \str{\aalph}$
  from $\aograph$ to the $\sbfunct$-coalgebra $\pair{\str{\aalph}}{\nbase{\kappa}}$ of all streams
  with respect to $\nbase{\kappa}$:
  
  \begin{center}
    \vspace{-.5ex}
    \begin{tikzpicture}[node distance=35mm,>=stealth,thick]
      \node (S) {$S$};
      \node (TS) [below of=S,node distance=12mm] {$\sum_{k=2}^{\infty} \aalph \times S^k$};
      \node (Aw) [right of=S,node distance=50mm] {$\str{\aalph}$};
      \node (TAw) [right of=TS,node distance=50mm] {$\sum_{k=2}^{\infty} \aalph \times (\str{\aalph})^k$};
      \draw [->] (S) -- (Aw) node [midway,above] {$\sinterpret$}; 
      \draw [->,shorten >= -1mm] (S) -- (TS) node [midway,left] {$\pair{o}{n}$};
      \draw [->,shorten >= -1mm] (Aw) -- (TAw) node [midway,right] {$\nbase{\kappa}$}; 
      \draw [->] (TS) -- (TAw) node [midway,above] {$\sum_{k=2}^{\infty} \text{id}\times\sinterpret^k$}; 
    \end{tikzpicture}
    \vspace{-.5ex}
  \end{center}
  The observation graph $\aograph$ \emph{defines} the stream $\interpret{\groot} \in \str{\aalph}$.
  A \emph{mix\nb-observation graph} is an $\nbase{\kappa}$\nb-observation graph for some~$\kappa$.
\end{definition}

The following result is a generalization of Theorem~\ref{thm:zip:auto}.
The key idea is to adapt Definition~\ref{def:zip:graph} 
by computing the derivatives $\nf{\pjx{0}{k}{t}},\ldots,\nf{\pjx{k-1}{k}{t}}$ 
of a zip\nb-term~$t$ 
where now~$k$~is~the arity of the first zip-symbol in the tree unfolding of~$t$.
Moreover, we note that mix\nb-DFAOs yield mix\nb-observation graphs
by collapsing states that generate the same stream
(for each of the equivalence classes one representative and its outgoing edges is chosen).
This collapse caters for mix\nb-DFAOs~which employ different bases for states that generate the same stream.

\begin{theorem}\label{thm:mix:zip:auto}
  For streams $\astr \in \str{\aalph}$ the following properties are equivalent:
  \begin{enumerate}
    \item The stream $\astr$ is mix-automatic.
    \item The stream $\astr$ can be defined by a zip-mix specification.
    \item There exists a finite mix-observation graph defining $\astr$. 
  \end{enumerate}
\end{theorem}

\begin{example}\label{ex:zip:mix:ex:sdb:dfao}
  The zip-mix specification corresponding
  to the mix-automaton from Example~\ref{ex:sdb:dfao} is:
  \begin{align*}
    \Zvar{0}  & = a \xcons \Zvar{0}' &
    \Zvar{0}' & = \zipn{2}{\Zvar{1},\Zvar{0}'} \\
    \Zvar{1}  & = b \xcons \Zvar{1}' &
    \Zvar{1}' & = \zipn{3}{\Zvar{0},\Zvar{1},\Zvar{2}'} \\
    \Zvar{2}  & = b \xcons \Zvar{2}' &
    \Zvar{2}' & = \zipn{2}{\Zvar{0},\Zvar{1}'}
  \end{align*}
\end{example}

We have seen that equivalence for zip-$k$ specifications is decidable (Theorem~\ref{thm:zipk:decide}),
and it can be shown that comparing zip-$k$ with zip-mix is decidable as well.
In the next section we show that equivalence becomes undecidable 
when zip-mix specifications are extended with projections $\spjx{i}{k}$.
But what about zip-mix specifications?
\begin{question}
  Is equivalence decidable for zip-mix specifications?
\end{question}

\section{Stream Equality is $\cpi{0}{1}$-complete}\label{sec:undecidable}

In this section, we show that the decidability results for the equality of
zip-$k$ specifications are on~the~verge
of undecidability.
To this end we consider an extension of the format of
zip\nb-specifications with the projections~$\spjx{i}{k}$.

\begin{definition}
  The set $\mcl{Z^\pi}(\aalph,\mcl{X})$ of \emph{zip$^\pi$\nb-terms} over $\pair{\aalph}{\mcl{X}}$
  is defined by the grammar:
  \vspace{-2.5ex}
  \begin{align*}
    Z & \BNFis \rv{X} \BNFor a \xcons Z \BNFor \zipn{k}{\overbrace{Z,\ldots,Z}^{\text{$k$ times}}} 
    \BNFor \pjx{i}{k}{Z}
  \end{align*}
  where $\rv{X} \in \mcl{X}$, $a \in \aalph$, $i,k \in \nat$.
  A \emph{zip$^\pi$-specification} 
  consists for every $\rv{X} \in \mcl{X}$ of an equation $\rv{X} = t$
  where $t \in \mcl{Z^\pi}(\aalph,\mcl{X})$.
\end{definition}

The class of zip$^\pi$-specifications forms a subclass 
of pure specifications~\cite{endr:grab:hend:isih:klop:2010},
and hence their productivity is decidable.
In contrast, the equivalence of zip$^\pi$\nb-specifications turns out to be undecidable
(even for productive specifications).

\begin{theorem}\label{thm:pi01}
  The problem of deciding the equality of streams 
  defined by productive zip$^\pi$\nb-specifications
  is $\cpi{0}{1}$-complete.
\end{theorem}

For the proof of the theorem, 
we devise a reduction from~the halting problem of Fractran programs (on the input~$2$)
to an equivalence problem of zip$^\pi$\nb-specifications.
Fractran~\cite{conw:1987} is a Turing-complete programming language.
As intermediate step of the reduction 
we employ an extension of Fractran programs with output (and immediate termination):

\newcommand{\sstepout}{\lambda}
\newcommand{\stepout}{\funap{\sstepout}}
\newcommand{\ssteps}[1]{\delta_#1}
\newcommand{\steps}[1]{\funap{\ssteps{#1}}}
\newcommand{\sout}[1]{\lambda^*_{#1}}
\newcommand{\out}[1]{\funap{\sout{#1}}}

\begin{definition}\label{def:fractran:function}
  An \emph{Fractran program with output} consists of:
  \begin{itemize}
    \item a list of fractions $\iafrac{1},\ldots,\iafrac{k}$ 
          ($k,p_1,q_1,{\ldots},p_k,q_k \,{\in}\, \pnat$),
    \item a partial \emph{step output} function $\sstepout : \{1,\ldots,k\} \pto \Gamma$
  \end{itemize}
  where $\Gamma$ is a finite output alphabet.
  A \emph{Fractran program} is a Fractran program with output
  for which $\undefd{\stepout{1}}, \ldots, \undefd{\stepout{k}}$.

  Let $\aprg$ be a Fractran program with output as above.
  Then~we define the partial function
  $\tuple{\cdot} \funin \nat \pto \{1,\ldots,k\}$ 
  that for every $n\in\nat$ selects the index $\tuple{n}$
  of the first applicable fraction by:
  \begin{align*}
    \tuple{n} = \min \;\{\,i \mid 1 \le i \le k,\; n\cdot \iafrac{i} \in \nat\,\}
  \end{align*}
  where we fix $\undefd{(\min \setemp)}$.
  We define $\fstep{\aprg} \funin \nat \to \nat \cup \Gamma \cup \{\bot\}$ by: 
  \begin{equation*}
    \funap{\fstep{\aprg}}{n} =
    \begin{cases}
      n\cdot\iafrac{\tuple{n}} & \text{if $\defd{\tuple{n}}$ and $\undefd{\stepout{\tuple{n}}}$}\\
      \stepout{\tuple{n}} & \text{if $\defd{\tuple{n}}$ and $\defd{\stepout{\tuple{n}}}$} \\
      \bot        & \text{if $\undefd{\tuple{n}}$}
    \end{cases}
  \end{equation*}
  for all $n \in \nat$.
  The first case is a \emph{computation step},
  the latter two are \emph{termination with} and \emph{without output}, respectively.

  We define the \emph{output function $\sout{\aprg} : \nat \pto \Gamma \cup \{\bot\}$ of $\aprg$} by
  \begin{align*}
    \out{\aprg}{n} =
      \begin{cases}
        \gamma & \text{if $\gamma = \funap{\fstep{\aprg}^{\hspace{.05em}i}}{n} \in \Gamma \cup \{\bot\}$ for some $i \in \nat$} \\
        \undefd{} & \text{if no such $i$ exists}
      \end{cases}
  \end{align*}
  If $\defd{\out{\aprg}{n}}$ then $\aprg$ is said to 
  \emph{halt on $n$ with output $\out{\aprg}{n}$}.
  Then $\aprg$ is called \emph{universally halting}
  if $\aprg$ halts on every $n \in \pnat$,
  and $\aprg$ is \emph{decreasing} if $\ianum{i} < \iaden{i}$ for every $1 \le i \le k$
  with $\undefd{\stepout{i}}$.
\end{definition}

For convenience, we denote Fractran programs with output by
lists of annotated fractions
where $\undefd{\stepout{i}}$ is represented by the empty word (no annotation):
\begin{align*}
  \iafrac{1} \stepout{1},\ldots,\iafrac{k} \stepout{k}
\end{align*} 

\begin{lemma}[\cite{conw:1987}]\label{lem:fractran:pi01}
  The problem of deciding on the input of a Fractran program
  whether it halts on $2$ is $\csig{0}{1}$-complete.
\end{lemma}

\newcommand{\halt}[1]{\chi_{#1}}
\newcommand{\halta}{\halt{a}}
\newcommand{\haltb}{\halt{b}}
We transform Fractran programs $\aprg$ into two decreasing (and therefore universally halting) 
Fractran programs $\aprg_0$ and $\aprg_1$ with output
such that $\aprg$ halts on input $2$
if and only if there exists $n\in \nat$ such that 
the outputs of $\aprg_0$ and $\aprg_1$ differ on $n$.

\begin{definition}\label{def:f0f1}
  Let $\aprg = \prglist{\iafrac{1},\ldots,\iafrac{k}}$ be a Fractran program.
  Let $a_1 < \ldots <a_m$ be the primes occurring in the factorizations of 
  $\ianum{1},\ldots,\ianum{k},\iaden{1},\ldots,\iaden{k}$.
  Let $z_1, z_2, c$ be primes such that 
  $z_1,z_2,c > \prod_{0\le i\le k} p_i\cdot q_i $, and
  $z_1 > z_2$ and $z_1 > 2\cdot c$.

  We define the Fractran program $\aprg^0$ with output as:
  \begin{align*}
    \overbrace{\frac{\ianum{1}}{\iaden{1}\cdot z_2},\;\ldots,\;\frac{\ianum{k}}{\iaden{k}\cdot z_2}}^{\text{simulate $\aprg$}},\;\;\;
    \overbrace{\frac{1}{a_1},\;\ldots,\;\frac{1}{a_m}}^{\text{cleanup}},\\
    \underbrace{\frac{1}{c\cdot z_2}\halta}_{\text{$\aprg$ halted}},\; \frac{1}{c},\;
    \underbrace{\frac{z_2}{z_1\cdot z_1},\; \frac{2\cdot c}{z_1}}_{\text{initialization}},\;
    \underbrace{\frac{1}{1}\haltb}_{\text{$\aprg$ did not halt}}
  \end{align*}
  Let $\aprg^1$ be obtained from $\aprg^0$ by dropping
  $\frac{z_2}{z_1\cdot z_1}$ and $\frac{2\cdot c}{z_1}$.
\end{definition}

\begin{lemma}\label{lem:total}
  The programs $\aprg^0$, $\aprg^1$ are decreasing and universally halting,
  and $\out{\aprg^i}{n} \in \{\halta,\haltb\}$ for all $n\in\nat$, $i \in \{0,1\}$.
\end{lemma}

\begin{lemma}\label{lem:ff0f1}
  The following statements are equivalent:
  \begin{enumerate}
    \item $\out{\aprg^0}{n} = \out{\aprg^1}{n}$ for all $n \in \pnat$.
    \item $\out{\aprg^0}{z_1^{e_1} \cdot z_2^{e_2}} = \out{\aprg^1}{z_1^{e_1} \cdot z_2^{e_2}}$ for all $e_1,e_2 \in \nat$.
    \item The Fractran program $\aprg$ does not halt on $2$.
  \end{enumerate}
\end{lemma}

Next, we translate Fractran programs to zip$^\pi$-specifications.
\begin{definition}\label{def:frac:spec}
  Let $\aprg = \iafrac{1}\stepout{1},\ldots,\iafrac{k}\stepout{k}$ be a decreasing Fractran program with output.
  
  Let $d \defdby  \lcm{\iaden{1},\ldots,\iaden{k}}$,
  and define 
  $\anum_n' = d \cdot \ianum{\tuple{n}}/\iaden{\tuple{n}}$ and 
  $\aoff_n = n\cdot\ianum{\tuple{n}}/\iaden{\tuple{n}}$
  for $1 \le n \le d$;
  if $\undefd{\tuple{n}}$, let $\undefd{\anum_n'}$ and $\undefd{\aoff_n}$.
  We define the \emph{zip$^\pi$-specification $\aspec(\aprg)$}
  for $1 \le n \le d$ by:
  \begin{align*}
    \Zroot &= \zipn{d}{\Zvar{1},\ldots,\Zvar{d}}\\
    \Zvar{n} &= \pjx{\aoff_n-1}{\anum_n'}{\Zvar{0}}
      &&\text{if $\defd{\tuple{n}}$ and $\undefd{\stepout{\tuple{n}}}$} \\
    \Zvar{n} &= \cons{\stepout{\tuple{n}}}{\Zvar{n}}
      &&\text{if $\defd{\tuple{n}}$ and $\defd{\stepout{\tuple{n}}}$} \\
    \Zvar{n} &= \cons{\bot}{\Zvar{n}}
      &&\text{if $\undefd{\tuple{n}}$}
  \end{align*}
\end{definition}

\begin{lemma}\label{lem:frac:spec}
  Let $\aprg$ be a decreasing Fractran program  with output.
  The zip$^\pi$-specification $\aspec(\aprg)$ 
  is productive and it holds that
  $\interpret{\Zroot}^{\aspec(\aprg)}(n) = \out{\aprg}{n+1}$ for every $n \in \nat$.
\end{lemma}

\begin{proof}[Proof of Theorem~\ref{thm:pi01}]
  We reduce the \emph{complement} of the halting problem of Fractran programs on input 2 
  (which is $\cpi{0}{1}$\nb-complete by Lemma~\ref{lem:fractran:pi01})
  to equivalence of zip$^\pi$-specifications.

  Let $\aprg$ be a Fractran program.
  Define $\aprg^0$, $\aprg^1$ as in Definition~\ref{def:f0f1}.
  By Lemma~\ref{lem:total} both are decreasing.
  By Lemma~\ref{lem:frac:spec} $\aspec(\aprg^i)$ is productive, and
  we have $\interpret{\Zroot}^{\aspec(\aprg^i)}(n) = \out{\aprg^i}{n+1}$ for every $n\in\nat$ and $i \in \{0,1\}$.
  Finally, by Lemma~\ref{lem:ff0f1} it follows that $\aspec(\aprg^0)$ and $\aspec(\aprg^1)$
  are equivalent if and only if $\aprg$ halts on $2$.
  
  The equivalence problem of productive specifications is obviously
  in $\cpi{0}{1}$ since every element can be evaluated.
\end{proof}


\paragraph*{Related work}

The complexity of deciding the equality of streams defined
by systems of equations 
has been considered in~\cite{rosu:2006} and~\cite{bale:2010}.
In~\cite{rosu:2006}, Ro\c{s}u shows $\cpi{0}{2}$\nb-completeness
of the problem for (unrestricted) stream equations.
In~\cite{bale:2010}, Balestrieri strengthens the result to
polymorphic stream equations.
However, both results depend on the use of ill-defined (non-productive)
specifications that do not uniquely define a stream.
The $\cpi{0}{2}$\nb-hardness proofs employ stream specifications 
for which productivity coincides with unique solvability.
%
As a consequence, both results depend crucially on the notion of
equivalence for specifications without unique solutions.

In contrast to \cite{rosu:2006} and~\cite{bale:2010}, 
we are concerned with productive specifications, that is, 
every element of which can be evaluated constructively.
Then equality is obviously in $\cpi{0}{1}$.
We show that equality is $\cpi{0}{1}$-hard
even for a restricted class of polymorphic, productive stream specifications.
%


\bibliography{main}

\newpage

\appendix
\section{Appendix}

\begin{proof}[Proof of Lemma~\ref{lem:free:root}]
  Introduce a fresh root $\rootsc'$ and add the equation $\rootsc' = \rootsc$.
\end{proof}

\begin{definition}
  A zip-specification $\mcl{S}$ is called \emph{$\szip$\nb-guarded} 
  if every cycle in $\mcl{S}$ 
  contains an occurrence of $\szip$.
\end{definition}

\begin{example}
  Specification~\eqref{eq:spec:morse} of the Thue--Morse stream 
  is $\szip$\nb-guarded.
  In contrast, the zip-specification:
  \begin{align*}
    \alts  & = \zip{\zeros}{\ones} &
    \zeros & = 0 \xcons \zeros &
    \ones  & = 1 \xcons \ones\,,
  \end{align*}
  specifying the stream
  $0 \xcons 1 \xcons 0 \xcons 1 \xcons 0 \ldots$ 
  of alternating zeros and ones,  
  is not $\szip$\nb-guarded.
\end{example}

It is an easy exercise to show that
$\szip$-free cycles correspond to periodic sequences,
and periodic sequences can be specified by zip-guarded zip-$k$ specification (for arbitrary $k$).
Hence:
\begin{lemma}\label{lem:zipguarded}
  Every zip-$k$ specification can be transformed
  into an equivalent, zip-guarded zip-$k$ specification.
\end{lemma}
\begin{proof}
  Every zip-free cycle $M = \ldots = c_1 : \ldots : c_n : M$
  characterizes a periodic sequence $\sigma = uuu\ldots$ with $u \in \aalph^*$.
  Note that for every $i,k \in \nat$, $\pjx{i}{k}{\sigma}$ is again periodic
  with a period length $\le \lstlength{u}$. 
  Thus the $\nbase{k}$-observation graph is finite (for every $k$)
  and hence by Lemma~\ref{lem:nbase:zip}
  we have a zip-guarded specification for~$\sigma$.
\end{proof}

\bigskip
\begin{proof}[Proof of Lemma~\ref{lem:flat}]
  Let $\aspec$ be a zip-specification.
  Using Lemma~\ref{lem:zipguarded} let $\aspec$ be zip-guarded.
  If $\aspec$ contains an equation of the form:
  \begin{align}
    \rv{X} = c_1 : \ldots : c_m : \zipn{k}{s_1,\ldots,s_i,\ldots,s_k}
    \tag{$*$}
  \end{align}
  such that $s_i \not\in \vars$ for some $1 \le i \le k$,
  then we pick a fresh~$\rv{X}'$ and
  replace the equation by:
  \begin{align*}
    \rv{X}  & = c_1 : \ldots : c_m : \zipn{k}{s_1,\ldots,\rv{X}',\ldots,s_k}\\
    \rv{X}' & = s_i
  \end{align*}
  Clearly, the resulting specification is equivalent to the original,
  and still zip-guarded.
  We repeat this transformation step until there are no equations of
  form ($*$) left, that is, the arguments of every occurrence of $\szipn{k}$-symbols
  are only recursion variables.
  
  Next, we replace equations of the form:
  \begin{align}
    \rv{X} &= c_1 : \ldots : c_m : \rv{Y}
    \tag{$\dagger$}
  \end{align}
  with $\rv{Y}$ a recursion variable, by (unfolding $\rv{Y}$):
  \begin{align*}
    \rv{X} &= c_1 : \ldots : c_m : r
  \end{align*}
  where the defining equation for $\rv{Y}$ is $\rv{Y} = r$.
  The obtained specification is equivalent and remains zip-guarded.
  Again, we repeat this step until there no longer are equations of form~($\dagger$).
  This process is guaranteed to terminate since the specification is zip-guarded.
  
  In the final specification, every right-hand side of an equation
  contains a $\szipn{k}$ (for some $k \ge 2$), and 
  this $\szipn{k}$ is applied to recursion variables only.
  Hence the final specification is flat.
  
  Furthermore, note that the resulting specification contains only $\szipn{k}$-symbols
  for $k \ge 2$ for which a $\szipn{k}$ also occurs in the original specification.
  As a consequence, the transformation preserves zip-$k$ specifications.
\end{proof}

\bigskip
\begin{proof}[Proof of Lemma~\ref{lem:interpret:unique}]
  Let $\aograph = \pair{S}{\pair{o}{n}}$ be a $\base$-observation graph,~and 
  let $\sinterpret_1, \sinterpret_2 \funin S \to \str{\aalph}$ be two $\safunct$-homomorphisms
  from $\aograph$ to $\mcl{S}_{\base} = \pair{\str{\aalph}}{\base}$.
  Define the relation $R \subseteq \str{\aalph} \times \str{\aalph}$ by
  $\interpret{t}_1 \mathrel{R} \interpret{t}_2$ for all $t \in S$.
  It is easy to check that $R$ is a $\base$\nb-bisimulation,
  and hence $\interpret{s}_1 = \interpret{s}_2$ for all $s \in S$ by Lemma~\ref{lem:coinduction}.
\end{proof}

\bigskip
\begin{proof}[Proof of Proposition~\ref{prop:ok=ok}]
  The final coalgebra for $\safunct$ is the $k$\nb-automaton
  of the $\aalph$\nb-weighted languages (or of the $k$\nb-ary trees with labels in $\aalph$)
  $\tuple{\aalph^{\kset^*},\pair{o}{n}}$
  where $\kset \defdby \setexp{0,1,\ldots,k}$,
  $ o \funin \aalph^{\kset^*} \to \aalph$, $L \mapsto \funap{L}{\emptyword}$,
  and 
  $ n \funin \aalph^{\kset^*} \to (\aalph^{\kset^*})^k$,
  $ \funap{n}{L} \mapsto \tuple{L_1,\ldots,L_k}$
  with $L_i$ defined by $\funap{L_i}{w} = \funap{L}{i\cdot w}$ for all $w\in\kset^*$
  (for the case $k=2$ see \cite[Thm.~3]{kupk:rutt:2011}).
  %
  The function
  $ f \funin \kset^* \to \omega $
  defined by
  $ \emptyword \mapsto 0 $, and
  $ a_0 a_1 \ldots a_n  \mapsto \sum_{j=0}^n (1 + a_j) k^j $
  is bijective and 
  induces an isomorphism from $\pair{\str{\aalph}}{\obase{k}}$ to $\tuple{\aalph^{\kset^*},\pair{o}{n}}$,
  because it has the property that $ \funap{f}{i\cdot w} = k \funap{f}{w} + i$ holds for all $i\in\kset$ and $w\in\kset^*$.
\end{proof}

\bigskip
\begin{proof}[Proof of Lemma~\ref{lem:finite:nk}]
  The equations of $\aspec$ are of the form:
  \begin{align*}
    \Zvar{j} = \cons{c_{j,0}}{\cons{\ldots}{\cons{c_{j,m_j-1}}{\zipn{k}{\Zvar{j,0},\ldots,\Zvar{j,k-1}}}}}
    && (0 \le j < n)
  \end{align*}
  We define $m \defdby \max\, \{\,m_i \mid 0 \le i < n \,\}$.
  
  %
  Note that the root $X_0$ is of the claimed form.
  Thus it suffices that 
  $\nf{\pjx{i}{k}{s}}$ is the shape
  whenever $s \in S$ is and $0 \le i < k$.
  Let $s = \cons{d_0}{\cons{\ldots}{\cons{d_{\ell-1}}{\Zvar{j}}}} \in S$ 
  with $\ell \le m$, and let $0 \le i < k$.
  Then it holds that:
  \begin{align*}
    \pjx{i}{k}{s}
    &\mred \underbrace{d_{i} : d_{i+k} : \ldots : d_{i+ak} :}_{\text{abbreviate as $D$}}\; \pjx{i'}{k}{X_j}\\
    &\mred D[\underbrace{c_{j,i'} : c_{j,i'+k} : \ldots : c_{j,i'+bk} :}_{\text{abbreviate as $C$}}\; X_{j,i''}]
  \end{align*}
  where $a$, $b$, $i'$ and $i''$ are defined by:
  \begin{align*}
    a &= \floor{(\ell-1-i)/k} & i' &= \modulo{k}{\ell-1-i} \\
    b &= \floor{(m_j-1-i')/k} & i'' &= \modulo{k}{m_j-1-i'} 
  \end{align*}
  The number of elements in $D[C[\Box]]$
  is at most $\floor{(\ell + m_j + 1)/k}$
  (since $\spjx{i}{k}$ `walks' over $\ell + m_j$ elements).
  Hence $D[C[X_{j,i''}]]$ is a normal form of the claimed form
  with a stream prefix of length $\le \floor{(\ell + m_j + 1)/k} \le \floor{(2m + 1)/k} \le m$.
\end{proof}

\begin{lemma}\label{lem:psize}
  For every productive, flat \zspec{}~$\aspec$ with root~$\rootsc$ it holds that:
  %
  \begin{equation*}
    \setsize{\synder{\aspec}{\Zvar{0}}} \le 2\cdot (\setsize{\Sigma} + 1)\cdot m \cdot n + 4\cdot m \punc{,}
  \end{equation*}
  where $n$ is the number of recursion variables and $m$ the longest prefix
  (as defined in the proof of Lemma~\ref{lem:finite:nk}).
\end{lemma}
\newcommand{\osder}{\funap{\delta}}
\newcommand{\osdern}[1]{\funap{\delta^{#1}}}
\begin{proof}
  We use the notation from the proof of Lemma~\ref{lem:finite:nk}.
  Here we consider only case $\nbase{2} = \pair{\seven}{\sodd}$.
  The proof of the general case works analogous.
  We define
  \begin{align*}
    T = \{\cons{c_1}{\cons{\ldots}{\cons{c_k}{\Zvar{i}}}} \where k \le m,\; c_i \text{ data elements},\; i \le n\}
  \end{align*}
  We have $\synder{\aspec}{\Zvar{0}} \subseteq T$ by the proof of Lemma~\ref{lem:finite:nk}.

  For a set $S$ of terms we define the one step derivatives $\osder{S}$ as follows:
  \begin{align*}
    \osder{S} = \{\nf{\even{s}},\,\nf{\odd{s}} \mid s \in S\}
  \end{align*}
  We consider a term $t \in T$:
  \begin{align*}
    t = \cons{c_1}{\cons{\ldots}{\cons{c_k}{\Zvar{i}}}}
  \end{align*}
  with $k \le m$. Then $\nf{\even{t}}$ and $\nf{\odd{t}}$ are of the form:
  \begin{align}
    \cons{c_{i_1}}{\cons{\ldots}{\cons{c_{i_{k'}}}{t'}}} \quad \text{with} \quad t' \in \osder{\Zvar{i}}\label{eq:osder}
  \end{align}
  where $k' \le \ceil{k/2}$ and $c_{i_1},\ldots,c_{i_{k'}}$ are residuals of $c_1,\ldots,c_k$.

  Then by induction it follows that every term $s \in \osdern{\ceil{\log_2 k}}{t}$ has the form:
  \begin{align}
    s = t' \quad\text{or}\quad s = \cons{c_{j}}{t'} \quad \text{with} \quad t' \in \osdern{\ceil{\log_2 k}}{\Zvar{i}}\label{eq:osderlog}
  \end{align}
  with $1 \le j \le k$.
  Furthermore, for $l\ge k$,
  every $s \in \osdern{\ceil{\log_2 l}}{t}$ has the form \eqref{eq:osderlog} 
  with $j\in\{1,\ldots, k\}$.
  Hence it follows:
  \begin{equation*}
    \osdern{\ceil{\log_2 m}}{T} 
      \subseteq
    \bigcup_{i=0}^{n-1} 
      \biggl(\,
        \osdern{\ceil{\log_2 m}}{\Zvar{i}}
          \cup
        \bigcup_{c\in\Sigma} \cons{c}{\osdern{\ceil{\log_2 m}}{\Zvar{i}}}
      \,\biggr)  
  \end{equation*}
  and therefore:
  \begin{equation*}
    \setsize{\osdern{\ceil{\log_2 m}}{T}} 
      \,\le\,
    (\setsize{\Sigma} + 1) \cdot
    \setsize{   
      \bigcup_{i=0}^{n-1} 
        \osdern{\ceil{\log_2 m}}{\Zvar{i}}
             }
  \end{equation*}
  Since $\setsize{\osdern{\ceil{\log_2 m}}{\Zvar{i}}} = 2^{\ceil{\log_2 m}}$
  we get:
  \begin{align*}
    \setsize{\osdern{\ceil{\log_2 {m}}}{T}} 
    & \;\le\; (\setsize{\Sigma} + 1) \cdot 2^{\ceil{\log_2 m}} \cdot n \\
    & \;\le\; (\setsize{\Sigma} + 1) \cdot 2\cdot m \cdot n
  \end{align*}
  Let $\osdern{\le i}{t} = \bigcup_{j \le i} \osdern{j}{t}$.
  Then $\synder{\aspec}{\Zvar{0}} \subseteq \osdern{\le \ceil{\log_2 m}}{\Zvar{0}} \cup \osdern{\ceil{\log_2 k}}{T}$
  and hence:
  \begin{align*}
    \setsize{\derivatives{\aspec}{\Zvar{0}}} 
    &\;\le\;
    \setsize{\osdern{\le \ceil{\log_2 m}}{\Zvar{0}}} + \setsize{\osdern{\ceil{\log_2 k}}{T}}\\
    &\;\le\;
    (1 + 2 + 4 + \ldots + 2^{\ceil{\log_2 m}}) + (\setsize{\Sigma} + 1) \cdot 2 m n\\
    &\;\le\;
    4\cdot m + (\setsize{\Sigma} + 1) \cdot 2\cdot m \cdot n
  \end{align*}

  \vspace{-4ex}
\end{proof}

\bigskip

\begin{proof}[Proof of Proposition~\ref{prop:flat:ptime}]
  By Lemmas~\ref{lem:finite:nk} and~\ref{lem:psize} the observation graphs of productive
  and flat \zspec{s} are finite and have polynomial size.
  By Proposition~\ref{prop:bisim} we can decide bisimilarity
  of observation graphs in linear time.
  Consequently equality of \zspec{s} is decidable in polynomial time.
\end{proof}

\bigskip

\begin{proof}[Proof of Proposition~\ref{prop-characterization}]
  First, we have a construction that works for any model $M$.
  For every point $a\in M$ and every $h$,
  we define the formula $\phi^h_{M,a}$. 
  The definition is by recursion on $h$
  (simultaneously for all $x\in M$) as follows: 
  $\phi^0_{M,a}$ is the conjunction of all atomic propositions
  ($0$ or $1$) satisfied by $a$ and all negations
  of atomic propositions not satisfied by $a$. 
  Given $\phi^h_{M,b}$ for all $b\in M$, we define 
  \begin{align*}
    \phi^{h+1}_{M,a} =
    & \bigwedge_{(a,b) \in \seven} \!\! \dmd{\seven} \phi^h_{M,b}
      \ \andd \ [\seven] \bigvee_{(a,b)\in \seven} \!\! \phi^h_{M,b} \\
    & \andd \ \bigwedge_{(a,b)\in \sodd} \!\! \dmd{\sodd} \phi^h_{M,b}
      \ \andd \ [\sodd] \bigvee_{(a,b)\in \sodd} \!\! \phi^h_{M,b} \\
    & \andd \ \phi^0_{M,a}
   \end{align*}
  We always identify sentences up to logical equivalence.
  
  As $M$ ranges over all models and $a$ over the points of $M$,
  we call the sentences $\phi^h_{M,a}$ the \emph{canonical sentences of height~$h$}.
  
  \begin{lemma}\label{lemma-canpremain}
    The following hold:
    \begin{enumerate}
      
      \item\label{partone}
        For all $h$, there are only finitely many sentences $\phi^h_{M,a}$.
      
      \item\label{unique} 
        For every $h$, every world of every model satisfies
        a unique canonical sentence of height $h$.
      
      \item\label{part-bisim}
        If $R$ is a bisimulation relation between models $M$ and $N$,
        and if $a \mathrel{R} b$, then $\phi^h_{M,a} =  \phi^h_{N,b}$.
      
      \item\label{converse}
        If $M$ and $N$ are finitely branching models, 
        then the converse of part~\ref{part-bisim} holds: 
        the largest bisimulation between $M$ and $N$ is the relation $R$ defined by
        \begin{equation}
          a \mathrel{R} b \quad \mbox{iff} \quad \text{for all $h$, $\phi^h_{M,a} =  \phi^h_{N_b}$}
          \label{eq-1}
        \end{equation}

    \end{enumerate}
  \end{lemma}

  \begin{proof}
    For parts~\ref{partone} and~\ref{unique}, 
    see Proposition 2.5 and Lemma 2.6 of~\cite{moss:2007}.
    (The arguments there are for one modality, 
    but the results extend in an obvious way to our setting.)   
    Part~\ref{part-bisim} is a standard fact, 
    as is the ``Hennessy--Milner'' result in part~\ref{converse}.
  \end{proof}
  
  Now we return to Proposition~\ref{prop-characterization}.
  We fix a finite model $M$ and some point $x$ in it.  
  Let $R$ be the largest bisimulation on $M$.
  It is a general fact that $R$ is an equivalence relation, 
  and indeed this also follows from the characterization in~\eqref{eq-1}.
  For all $a$, $b$ in $M$, if $\nott (a\ R\ b)$, 
  then there is some natural number $h$ so that 
  $\phi^h_{M,a} \neq \phi^h_{M,b}$.  
  Since $M \times M$ is a finite set, 
  there is some fixed $h^*$ so that for all $a, b\in M$, 
  if $\nott (a\ R\ b)$, then $\phi^{h^*}_{M,a} \neq \phi^{h^*}_{M,b}$.
  The key consequence of this is that for all $a,b\in M$,
  \begin{equation}
    \text{if $\phi^{h^*}_{M, a} = \phi^{h^*}_{M,b}$, then also $\phi^{h^*+1}_{M,a} = \phi^{h^*+1}_{M,b}$}
    \label{eq-2}
  \end{equation}
  From this and Lemma~\ref{lemma-canpremain}, part~\ref{unique},
  it follows that for all $a, b\in M$,
  \begin{equation}
    \text{if $a\models\phi^{h^*}_{M, b} $,  then also $a \models \phi^{h^*+1}_{M,b}$}
    \label{eq-3}
  \end{equation}
  For $a\in M$, let $\psi_a$ denote the formula
  \begin{align*}
    \phi^{h^*}_{M,a} \iif 
    \biggl(\; 
      & \bigwedge_{(a,b)\in \seven} \!\! \dmd{\seven} \phi^{h^*}_{M,b} \ 
        \andd \ [\seven] \bigvee_{(a,b)\in \seven} \!\! \phi^{h^*}_{M,b} \\
      & \andd \bigwedge_{(a,b)\in \sodd} \!\! \dmd{\sodd} \phi^{h^*}_{M,b} \ 
        \andd \ [\sodd] \bigvee_{(a,b)\in \sodd} \!\! \phi^{h^*}_{M,b} 
    \;\biggr)
   \end{align*}
  Using \eqref{eq-3}, we see that for all $a, b\in M$, $a\models\psi_b$.
  
  We now finish the proof of Proposition~\ref{prop-characterization}.
  We have our model $M$ and a point $x \in $M.
  We take the characterizing sentence of $x$ in $M$ to be 
  \begin{equation*}
    \phi_x \quadeq \phi^{h^*}_{M,x} \andd [(\seven \sqcup \sodd)^*] \bigwedge_{a\in M} \! \psi_a
  \end{equation*}
  It is easy to see that $x\models \phi_x$.
  To end our proof, suppose that $N$ is any model, and $y\in N$ satisfies $\phi_x$.
  We define a bisimulation $R$ between $M$ and $N$ which relates $x$ to $y$:
  \begin{align*} 
    a \mathrel{R} b 
    \quad\text{iff}\quad
    & \text{$a$ is reachable in $M$ from $x$ using  $(\seven\sqcup\sodd)^*$,} \\
    & \text{$b$ is reachable in $N$ from $y$ using  $(\seven\sqcup\sodd)^*$,} \\
    & \text{and $\phi^{h^*}_{M,a} = \phi^{h^*}_{N,b}$}
  \end{align*}
  The definition of $\psi_a$ ensures that $R$ is a bisimulation.
  
  This completes the proof.
\end{proof}
As an example of the construction in the above proof 
we obtain the following sentence~$\phi_{\con{M}}$ 
characterizing the Thue\nb--Morse sequence~$\con{M}$: 
Let $\boxdmd{\pi}\phi$ abbreviate $\dmd{\pi}\phi \andd \bx{\pi}\phi$,
and let
\begin{align*}
  \phi & = 0 \andd \nott 1 \andd  \boxdmd{\seven}0\andd \boxdmd{\sodd}1 \\
  \psi & = \nott 0 \andd  1 \andd \boxdmd{\seven}1 \andd \boxdmd{\sodd}0
\end{align*}
Then 
\begin{align*}
  \phi_{\con{M}}
  =
  \phi \andd [(\seven\sqcup\sodd)^*]
  \biggl( 
    & (\phi \iif \boxdmd{\seven}\phi \andd \boxdmd{\sodd}\psi) \\
    & \andd (\psi \iif \boxdmd{\seven}\psi \andd \boxdmd{\sodd}\phi) \\[1.5ex]
    & \andd (0 \iif \boxdmd{\seven}0) \\
    & \andd (1 \iif \boxdmd{\seven}1)
  \biggr)
\end{align*}

\begin{theorem}[Cobham~\cite{cobh:1969}]\label{thm:cobham}
  Let $k,\ell \ge 2$ be multiplicatively independent
  (i.e., $k^a \ne \ell^b$, for all $a,b \gt 0$),
  and let $\astr \in \str{\aalph}$ be 
  both $k$- and $\ell$\nb-automatic. 
  Then $\astr$ is eventually periodic.
\end{theorem}

\begin{proof}[Proof of Proposition~\ref{prop:mix:extends}]
  Let $k,\ell$ be multiplicatively independent integers,
  and let $\astr \in \str{\aalph_{\aaut}}$ and $\bstr \in \str{\aalph_{\baut}}$ be 
  $k$- and $\ell$\nb-automatic sequences, generated by DFAOs 
  $\aaut = \tuple{Q_{\aaut},\fin{k},\delta_{\aaut},q_{\aaut,0},\aalph_{\aaut},\lambda_{\aaut}}$
  and $\baut = \tuple{Q_{\baut},\fin{\ell},\delta_{\baut},q_{\baut,0},\aalph_{\baut},\lambda_{\baut}}$,
  respectively.
  Assume that both $\astr$ and $\bstr$ are not eventually periodic.
  
  We show that the sequence $\cstr = \zip{\astr}{\bstr} \in \str{(\aalph_{\aaut} \cup \aalph_{\baut})}$ 
  is mix-automatic, but that there is no integer $m \ge 2$ such that $\cstr$ is $m$\nb-automatic.
  For the first, note that $\cstr$ is generated by the mix\nb-DFAO
  $\caut = \tuple{S,\{\fin{\beta(s)}\}_{s \in S},\delta,s_0,\aalph_{\aaut}\cup\aalph_{\baut},\lambda}$,
  where $S = \{s_0\} \cup Q_{\aaut} \cup Q_{\baut}$,
  and with $\beta$, $\delta$ and $\lambda$ defined by
  \begin{align*}
    \beta(s_0) & = 2    & \delta(s_0,0) & = q_0 & \lambda(s_0) &= \lambda_{\aaut}(q_0) \\
               &        & \delta(s_0,1) & = r_0 \\
    \beta(q)   & = k    & \delta(q,a)   & = \delta_{\aaut}(q,a) & \lambda(q) & = \lambda_{\aaut}(q) \\
    \beta(r)   & = \ell & \delta(r,b)   & = \delta_{\baut}(r,b) & \lambda(r) & = \lambda_{\baut}(r)
  \end{align*}
  for all $q \in Q_{\aaut}$, $r \in Q_{\baut}$, $a \in \fin{k}$ and $b \in \fin{\ell}$.
  Now suppose $\cstr$ is $m$\nb-automatic, for some $m \ge 2$.
  Then, as $m$\nb-automaticity is closed under arithmetic subsequences,
  also $\astr$ and $\bstr$ are $m$\nb-automatic.
  By Theorem~\ref{thm:cobham} and the assumption that $\astr,\bstr$ are not eventually periodic,
  it follows that $k$ and $m$ are not multiplicatively independent, and likewise for $\ell$ and $m$.
  But then $k$ and $\ell$ are not multiplicatively independent, contradicting our assumption.
\end{proof}

\begin{lemma}\label{lem:frac:decreasing}
  Every decreasing Fractran program with output is universally halting.
\end{lemma}

\begin{proof}[Proof of Lemma~\ref{lem:frac:decreasing}]
  Then
  $\funap{\fstep{\aprg}}{n} < n$ or $\funap{\fstep{\aprg}}{n} \in \Gamma \cup \{\bot\}$
  for every $n \in \nat$.
  Hence $\aprg$ is universally halting.
\end{proof}

\bigskip
\begin{proof}[Proof of Lemma~\ref{lem:total}]
  By the choice of $z_1$, $z_2$ and $c$ we have $q > p$~for~all 
  annotated-free fractions $\frac{p}{q}$ from $\aprg^0$ and $\aprg^0$.
  Hence $\aprg^0$ and $\aprg^0$ are decreasing and by Lemma~\ref{lem:frac:decreasing} universally halting.
  
  Since $\aprg^0$ and $\aprg^1$ contain $\frac{1}{1}\haltb$,
  it follows that the programs always terminate with output,
  that is, either $\halta$ or $\haltb$.
\end{proof}

\bigskip
The following lemmas use the notation from Definition~\ref{def:f0f1}:

\begin{lemma}\label{lem:2cz2}
  For every $n\in \nat$ 
  we have $\out{\aprg^0}{2 \cdot c \cdot z_2^{n}} = \halta$ 
  if and only if the Fractran program $\aprg$ halts on input $2$ within $n$ steps
  (that is, $\myex{n' \le n}{\funap{\fstep{\aprg}^{n'}}{2}} = \bot$).
\end{lemma}

\begin{proof}[Proof of Lemma~\ref{lem:2cz2}]
  On inputs of the form $2 \cdot c \cdot z_2^{n}$ the program $\aprg^0$
  behaves as $\aprg$ except for removing a prime $z_2$ in each step.
  By definition of $\aprg^0$ and induction we obtain
  \begin{align*}
    \fstep{\aprg^0}^i(2 \cdot c \cdot z_2^{n})
    = \fstep{\aprg}^i(2) \cdot c \cdot z_2^{n-i}
  \end{align*} 
  for every $i \le n$ such that $\myall{j \le i}{\funap{\fstep{\aprg}^j}{2}} \ne \bot$.

  Assume that there is $0 <n' \le n$ such that
  $\funap{\fstep{\aprg}^{n'-1}}{2} \in \nat$ and $\funap{\fstep{\aprg}^{n'}}{2} = \bot$, that is,
  $\aprg$ halts in precisely~$n'$ steps.
  Then in $\aprg^0$ to the value 
  $\fstep{\aprg^0}^{n'-1}(2 \cdot c \cdot z_2^{n}) = \fstep{\aprg}^{n'-1}(2) \cdot c \cdot z_2^{n-(n'-1)}$
  none of the `simulate $\aprg$'-fractions is applicable,
  and after the `cleanup'\nb-fractions have removed all primes occurring in $\aprg$,
  the fraction $\frac{1}{c\cdot z_2}\halta$ will result in termination with output $\halta$
  (the fraction is applicable since $n-(n'-1) > 0$).
  
  On the other hand assume $\funap{\fstep{\aprg}^{n'}}{2} \in \nat$ for all $n' \le n$,
  that is, $\aprg$ does not terminate within $n$ steps.
  Then $\fstep{\aprg^0}^n(2 \cdot c \cdot z_2^{n}) = \fstep{\aprg}^n(2) \cdot c$.
  As a consequence, in $\aprg^0$ the `simulate $\aprg$'-fractions and
  likewise $\frac{1}{c\cdot z_2}\halta$ are not applicable (lacking $z_2$),
  and thus the `cleanup'-fractions will remove all primes occurring in $\aprg$
  and then $\frac{1}{c}$ removes the $c$.
  Finally, only $\frac{1}{1}\haltb$ is applicable, resulting in termination with output $\haltb$.
\end{proof}

\begin{lemma}\label{lem:f0f1:coincide}
  It holds that
  $\fstep{\aprg^0}(n) = \fstep{\aprg^1}(n)$
  for all $n \in \nat$ such that
  any of $a_1\ldots,a_m$ or $c$ divides $n$.
\end{lemma}

\begin{proof}[Proof of Lemma~\ref{lem:f0f1:coincide}]
  For every $p \in \{a_1\ldots,a_m,c\}$
  there is a fraction $\frac{1}{p}$ in the common prefix of $\aprg^0$ and $\aprg^1$.
\end{proof}

\bigskip
\begin{proof}[Proof of Lemma~\ref{lem:ff0f1}]
  The equivalence of (i) and (ii) is 
  a consequence~of Lemma~\ref{lem:f0f1:coincide}
  since $z_1$ and $z_2$ are the only remaining primes from $\aprg^0$ and $\aprg^1$
  not covered by this lemma.
  
  Consider (ii) $\Rightarrow$ (iii).
  Let $n \in \nat$, we show that $\aprg$ does not halt on $2$ within $n$ steps.
  We have $\funap{\fstep{\aprg^0}}{z_1 \cdot z_2^n} = 2\cdot c \cdot z_2^n$
  and $\funap{\fstep{\aprg^1}}{z_1 \cdot z_2^n} = \haltb$,
  and thus $\out{\aprg^0}{2\cdot c \cdot z_2^n} = \haltb$ by (ii).
  We conclude this case with an appeal to Lemma~\ref{lem:2cz2}.
  
  For (iii) $\Rightarrow$ (ii), let $e_1,e_2 \in \nat$.
  We have $\out{\aprg^1}{z_1^{e_1} \cdot z_2^{e_2}} = \haltb$.
  Assume $e_1 = 2 \cdot n$, $n\in\nat$ then
  $\funap{\fstep{\aprg^0}^{n}}{z_1^{e_1} \cdot z_2^{e_2}} = z_2^{n+e_2}$
  and hence $\out{\aprg^0}{z_1^{e_1} \cdot z_2^{e_2}} = \haltb$ (by definition of $\aprg^0$).
  Otherwise $e_1 = 2 \cdot n + 1$ and we obtain
  $\funap{\fstep{\aprg^0}^{n}}{z_1^{e_1} \cdot z_2^{e_2}} = 2\cdot c \cdot z_2^{n+e_2}$.
  Then it follows $\out{\aprg^0}{z_1^{e_1} \cdot z_2^{e_2}} = \haltb$ by Lemma~\ref{lem:2cz2}.
\end{proof}

\bigskip
\begin{proof}[Proof of Lemma~\ref{lem:frac:spec}]
  Let $\sinterpret{}$ be a solution for $\aspec(\aprg)$.
  We use the notation from Definition~\ref{def:frac:spec}.
  We prove $$\interpret{\Zroot}^{\aspec(\aprg)}(n) = \out{\aprg}{n+1}$$
  by induction on $n\in\nat$.

  Let $n\in\nat$ and $m \in N_{<d}$, $o \in \nat$ such that $n= m + od$.
  For every $1 \le i \le k$: $(n+1) \cdot \iafrac{i} = (m+1)\cdot \iafrac{i} + od \cdot \iafrac{i}$.
  By choice of $d$: $d\cdot \iafrac{i} \in \nat$,
  hence $(n+1) \cdot \iafrac{i} \in \nat \Leftrightarrow (m+1) \cdot \iafrac{i} \in \nat$
  and thus  $\tuple{n+1} = \tuple{m+1}$ and $\stepout{\tuple{n+1}} = \stepout{\tuple{m+1}}$.

  We have $\interpret{X_0}(n) = \interpret{X_{m+1}}(o)$.
  If $\undefd{\tuple{n+1}}$ then $\undefd{\tuple{m+1}}$ and hence $\interpret{X_0}(n) = \interpret{X_{m+1}}(o) = \bot$.
  Likewise for $\defd{\tuple{n+1}}$ and $\defd{\stepout{\tuple{n+1}}}$ we obtain
  $\interpret{X_0}(n) = \stepout{\tuple{n+1}}$.
  
  The remaining case is $\defd{\tuple{n+1}}$ and $\undefd{\stepout{\tuple{n+1}}}$.
  Then 
  \begin{align*}
  \interpret{X_0}(n) 
  &= (\pjx{\aoff_{m+1}-1}{\anum_{m+1}'}{\interpret{X_0}})(o) \\
  &= \interpret{X_0}(\aoff_{m+1}-1 + \anum_{m+1}'\cdot o) \\
  &= \interpret{X_0}((m+1)\cdot \iafrac{\tuple{n+1}}-1 + d\cdot \iafrac{\tuple{n+1}}\cdot o) \\
  &= \interpret{X_0}((m+1+ do) \cdot \iafrac{\tuple{n+1}}-1) \\
  &= \interpret{X_0}(\fstep{\aprg}(n+1)-1) \\
  &= \out{\aprg}{\fstep{\aprg}(n+1)}  \quad \text{by induction hypothesis} \\
  &= \out{\aprg}{n+1}
  \end{align*}
  Finally, we note that since the derivations
  from $\interpret{X_0}(n)$ to $\interpret{X_0}(\fstep{\aprg}(n+1)-1)$
  also exist on the level of rewrite sequences in the zip-specification,
  and by decreasingness it holds $\fstep{\aprg}(n+1)-1 < n$,
  it follows that the specification is productive.
\end{proof}

\end{document}